\documentclass[
final
]{dmtcs-episciences}

\usepackage[utf8]{inputenc}
\usepackage[T1]{fontenc}
\usepackage{amsmath,amssymb,amsfonts,tikz,graphicx,mathtools}
\usepackage{cases,subfigure}
\usetikzlibrary{automata,positioning,arrows,shapes}

\newcommand\Z{\mathbb{Z}}
\newcommand\N{\mathbb{N}}

\newcommand\restr[2]{{
  \left.\kern-\nulldelimiterspace 
  #1 
  \vphantom{\big|} 
  \right|_{#2} 
  }}

\title{Freezing, Bounded-Change and Convergent Cellular Automata\thanks{This work has been supported by ECOS Sud -- CONICYT project C12E05}}
\author{Nicolas Ollinger\affiliationmark{1} \and Guillaume Theyssier\affiliationmark{2}}
\affiliation{
Univ. Orl\'{e}ans, INSA Centre Val de Loire, LIFO EA 4022, Orl\'{e}ans, France\\
Institut de Mathématiques de Marseille (Université Aix Marseille, CNRS, Centrale Marseille), France
}

\newtheorem{definition}{Definition}
\newtheorem{theorem}{Theorem}
\newtheorem{corollary}{Corollary}
\newtheorem{proposition}{Proposition}
\newtheorem{fact}{Fact}
\newtheorem{example}{Example}
\newtheorem{remark}{Remark}
\newtheorem{lemma}{Lemma}

\tikzset{>=latex}

\received{2019-09-01}

\revised{2020-02-13, 2021-04-22}

\accepted{2021-12-24}

\begin{document}
\publicationdetails{24}{2022}{1}{2}{5734}

\maketitle
\begin{abstract}
  This paper studies three classes of cellular automata from a computational point of view: freezing cellular automata where the state of a cell can only decrease according to some order on states, cellular automata where each cell only makes a bounded number of state changes in any orbit, and finally cellular automata where each orbit converges to some fixed point. Many examples studied in the literature fit into these definitions, in particular the works on cristal growth started by S. Ulam in the 60s. The central question addressed here is how the computational power and computational hardness of basic properties is affected by the constraints of convergence, bounded number of change, or local decreasing of states in each cell. By studying various benchmark problems (short-term prediction, long term reachability, limits) and considering various complexity measures and scales (LOGSPACE vs. PTIME, communication complexity, Turing computability and arithmetical hierarchy) we give a rich and nuanced answer: the overall computational complexity of such cellular automata depends on the class considered (among the three above), the dimension, and the precise problem studied. In particular, we show that all settings can achieve universality in the sense of Blondel-Delvenne-K\r{u}rka, although short term predictability varies from NLOGSPACE to P-complete. Besides, the computability of limit configurations starting from computable initial configurations separates bounded-change from convergent cellular automata in dimension~1, but also dimension~1 versus higher dimensions for freezing cellular automata. Another surprising dimension-sensitive result obtained is that nilpotency becomes decidable in dimension~ 1 for all the three classes, while it stays undecidable even for freezing cellular automata in higher dimension.
\end{abstract}

\section{Introduction and Formal Setting}

Cellular automata (CA for short) are a well-studied model appearing in different research areas under different points of view. It is widely used as a modeling tool of fundamental physical phenomena \cite{chopard} or high-level phenomena from other disciplines \cite{alekseyevskaya92,Nagel92,Fuentes}. It is a rich class of symbolic dynamical systems\cite{hedlund} extensively studied both from the topological \cite{kurkabook} and ergodic point of view \cite{Pivato09}. It is a computational model very close to Turing machines but with a massively parallel feature: this translates into a specific algorithmic complexity theory \cite{Terrier2012}, into universal computational power and existence of universal objects \cite{OllingerUnivhistory}, but also into the ubiquity of undecidability and computational hardness in most of the properties and problems concerning them \cite{kari94-2,delacourt11,KariOllinger08,ccca}. These points of view can have interactions. For instance, the computational power of several classes of CA defined by dynamical features inspired by physics has been studied \cite{MoritaHarao,MoreiraConserve,GajardoKM12}. 

Here we follow that line of research and focus on the overall computational complexity of a class of CA defined by a very simple dynamical property: point-wise convergence. More precisely, we study three classes which form a hierarchy by inclusion: \emph{freezing CA} (\textit{i.e.} CA which are locally decreasing according to some order on states), \emph{bounded-change CA} (\textit{i.e.} CA with a global bound on the number of state changes a cell can make in any orbit), and the general class of point-wise convergent CA. One of our inspiration is the profusion of examples in the literature, from the seminal works of S. Ulam \cite{ulam} to the recent and fast development of self-assembly models \cite{Winslow16}, which all share the convergence property: a larger and larger zone of the configuration gets frozen by the dynamics while changes continue outside the zone. On the other hand, some previous works explicitly studied the class of freezing CA \cite{GolOlThey15,BeckerMOT18} or bounded-change CA~\cite{vollmar81} and established that universal computation is possible in any dimension but ``slowed down'' by the bounded-change constraint in dimension 1. This was done by studying the short-term prediction problem on one hand, and by giving an explicit encoding of Minsky machines into such CA on the other hand. 

The present paper is an extension\footnote{In particular, we give complete detailed proofs of all results of~\cite{GolOlThey15}, some of which were only sketched in this preliminary paper.} of~\cite{GolOlThey15} with two new ingredients. First, we extend the benchmark problems used to evaluate computational complexity and include long-term topological reachability as in the notion of universality of~\cite{DelvenneKB06}, but also limit fixed points and their dependence on the initial configuration. Second, we consider the whole class of point-wise convergent CA through the same approach, which to our knowledge is completely new. Our main results show that the overall computational complexity varies both within the hierarchy of the three classes and with the dimension considered (see Section~\ref{sec:recap} for a synthetic view). In particular, convergent CA are generally more powerful than bounded-change CA, while freezing and bounded-change CA have the same overall computational complexity. However, we show that even the most constrained setting (freezing 1D CA) can achieve universality in the sense of~\cite{DelvenneKB06} (Theorem~\ref{thm:freezecyreach}). This is counterbalanced by the fact that such CA cannot produce uncomputable limit points starting from computable initial configurations (Theorem~\ref{thm:computablelimits}), while convergent CA in dimension 1 can (Theorem~\ref{thm:uncomputablelimit}). Concerning dimension sensitiveness, we show that various aspects are affected: the capacity to efficiently compute or the capacity to produce uncomputable limit points from computable initial configuration for bounded-change CA, and the decidability of nilpotency for all the three classes.

\paragraph{Organization of the paper.}
The classical definitions and formal setting for CA in general is recalled in the next paragraph. In section~\ref{sec:maindefs} we introduce the three classes of CA considered in this paper together with several examples. In section~\ref{sec:dynamics} we focus on dynamical properties of such CA including nilpotency. Section~\ref{sec:ccu} is devoted to the computational complexity of two canonical problems of CA theory: short-term prediction and long term reachability. In section~\ref{sec:limits} we study problems specific to convergent CA concerning limit fixed-point reached from a given configuration. Finally, section~\ref{sec:recap} gives a brief recap of our main results.

\paragraph{Formal setting.}
For any finite set $Q$ (the alphabet or set of states) and positive integer $d$ (the dimension), we consider the symbolic space ${Q^{\Z^d}}$, \textit{i.e.} the set of maps called \emph{configurations} giving a state from $Q$ to each position in $\Z^d$. The state of configuration ${c\in Q^{\Z^d}}$ at position ${z\in\Z^d}$ will be denoted either ${c(z)}$ or ${c_z}$.
\newcommand\zd{{\Z^d}}

\newcommand\ball[1]{\mathcal{B}({#1})}
For any ${n\in\N}$, we define $\ball{n}$ as the set of positions of ${\Z^d}$ of norm (maximum of absolute values of coordinates) at most $n$: 
\[\ball{n} = \{z\in\Z^d : \|z\|_\infty\leq n\}.\]
Then for any \emph{bounded configuration} ${u\in Q^{\ball{n}}}$ of \emph{radius $n$}, we define the \emph{cylinder set} $[u]$ centered on cell~0 by:
\[[u] = \{c\in Q^{\Z^d} : \forall z\in\ball{n}, c_z=u_z\}.\]

In particular, every state $h$ can be considered as a bounded configuration of radius $0$ inducing a cylinder $[h]$ of every configuration $c$ with $c_0=h$.

The cylinder sets can be chosen as a base of open sets of the space $Q^\zd$ endowing it with a compact (and totally disconnected) topology \cite{kurkabook}. Equivalently, it can be defined by the following metric: 
\[\delta(c,c') = 2^{-\min \{\|z\|_\infty : c_z\neq c'_z\}}\]
often called the Cantor metric.

A cellular automaton of dimension $d$ and state set $Q$ is a map $F$
acting on the set of configuration $Q^{\zd}$ in a continuous and
translation invariant way. Equivalently (Curtis-Lyndon-Hedlund theorem \cite{hedlund}), it can
be defined locally by a neighborhood $V$ (a finite subset of $\zd$) and a
local transition function ${f:Q^V\rightarrow Q}$ as follows:
\[\forall z\in\zd,\quad F(c)_z = f\bigl(\restr{c}{z+V}\bigr)\]
where ${\restr{c}{z+V}}$ is the map ${z'\in V\mapsto c_{z+z'}}$.
The \emph{radius} of $F$ is the smallest integer $r$ such that ${V\subseteq\ball{r}}$ where $V$ is some neighborhood for which there is a local map $f_V:Q^V\rightarrow Q$ defining $F$ as above. $F$ induces an action on finite patterns as follows. For any ${n\in\N}$ and any ${u\in Q^{\ball{n+r}}}$, ${F(u)}$ is the finite pattern ${v\in Q^{\ball{n}}}$ obtained by application of $f$ on $u$ at each position from $\ball{n}$, \emph{i.e.\/} such that 
\[\forall c\in[u], F(c)\in[v].\]
We will frequently use the so called von Neumann neighborhood which is the following subset of $\Z^2$: ${V=\{(x,y) : |x|+|y|\leq 1\}}$. 

\newcommand\PTIME{\mathrm{PTIME}}
\newcommand\NL{\mathrm{NLOGSPACE}}
\newcommand\LOG{\mathrm{LOGSPACE}}

Finally, we will use the following standard complexity classes:
\begin{itemize}
\item $\PTIME$ is the set of problems which can be solved by a deterministic Turing machine in polynomial time;
\item $\LOG$ is the set of problems which can be solved by a deterministic Turing machine in logarithmic space;
\item $\NL$ is the set of problems which can be solved by a non-deterministic Turing machine in logarithmic space.
\end{itemize}

Without explicit mention, and in particular when speaking about P-completeness, we consider LOGSPACE reductions.

\section{Main Definitions, with Examples}
\label{sec:maindefs}

\begin{definition}
\label{def:freezing}
A CA $F$ is a \emph{freezing CA} if, for some (partial) order $\leq$ on states, the state of any cell can only decrease, \textit{i.e.\/}
  \[F(c)_z\leq c_z\]
  for any configuration $c$ and any cell $z$.
\end{definition}

\newcommand\screl[1]{\rightarrow_{#1}}
For any $F$ and states ${q\neq q'}$, denote by ${q \screl{F} q'}$ the fact that in some context a cell can change from state $q$ to $q'$ in one step. If $F$ is a freezing CA, then the transitive and reflexive closure of ${\screl{F}}$ is a (partial) order satisfying the condition of Definition~\ref{def:freezing}. Conversely, if ${\screl{F}}$ is acyclic then $F$ is freezing for the order given by the reflexive and transitive closure of ${\screl{F}}$. This gives a polynomial time algorithm (in the size of the transition table, the list of every possible output of a local map defining $F$) to test whether a given CA is freezing: build $\screl{F}$ and test its acyclicity.

\begin{fact}
  \label{fact:freezingdecidable}
  There is an algorithm to decide whether a given CA is freezing that runs in time ${O(n^{|V|})}$ where $n$ is its number of states and $V$ is its neighborhood.
\end{fact}

Our purpose is to study freezing cellular automata as a class and give general results for it, as started in \cite{GolOlThey15}. However, many particular freezing cellular automata have already been considered in the literature. We give below several examples showing the variety of possible behaviors within this class (see Figure~\ref{fig:examples_freezing}). 

\begin{figure}[tp]
  \centering
  \subfigure[][]{%
    \label{fig:ex3-a}%
    \includegraphics[width=3cm]{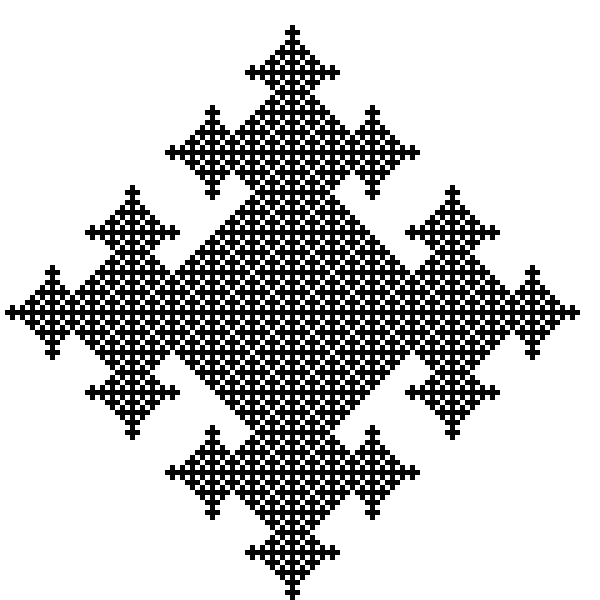}}%
  \hspace{8pt}%
  \subfigure[][]{%
    \label{fig:ex3-b}%
    \includegraphics[width=3cm]{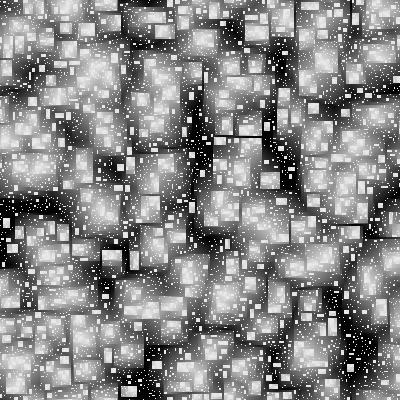}}%
  \hspace{8pt}%
  \subfigure[][]{%
    \label{fig:ex3-d}%
    \includegraphics[width=3cm]{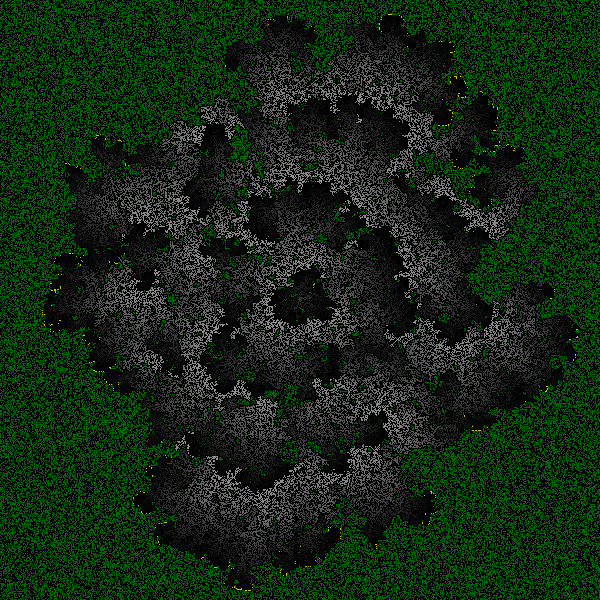}}%
  \hspace{8pt}%
  \subfigure[][]{%
    \label{fig:ex3-c}%
    \includegraphics[width=3cm]{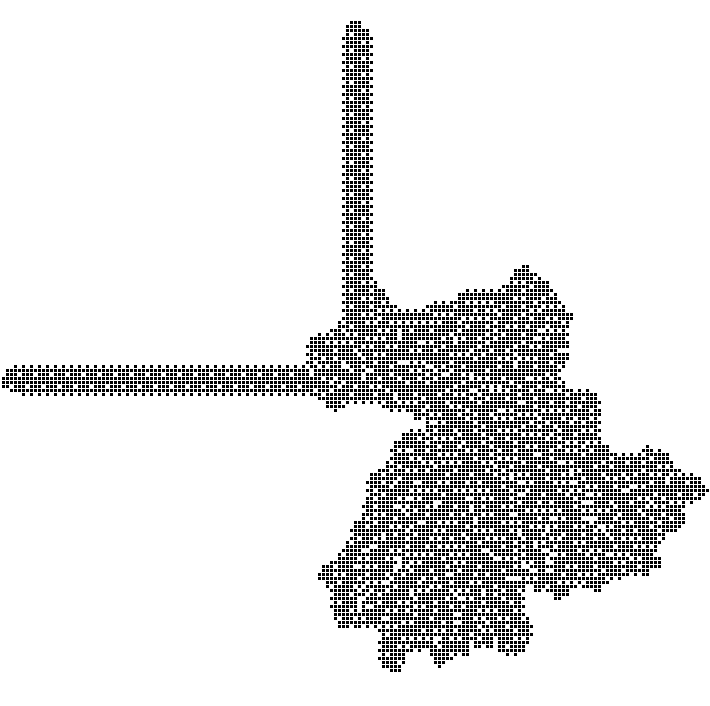}}%
  \caption[A set of four subfigures.]{Some examples of freezing CA studied in the literature:
    \subref{fig:ex3-a} configuration of Ulam's rule from Example~\ref{ex:ulam} after some steps starting from a single 1 in a sea of 0s;
    \subref{fig:ex3-b} spatio-temporal representation of the standard bootstrap percolation model (Example~\ref{ex:bootstrap}) where shades of gray represent time at which a cell turned to 1 (the whiter the sooner);
    \subref{fig:ex3-d} a configuration of a model of forest fire as in Example~\ref{ex:sir} (fire is represented in yellow, trees in green, and empty spaces in black and hashes in shades of gray indictaing the time modulo some larger integer at which the fire occupied that position for last time);
    \subref{fig:ex3-c} typical configuration of rule 'life without death' from Example~\ref{ex:lifewithoutdeath}}
  \label{fig:examples_freezing}
\end{figure}

\begin{example}
  \label{ex:ulam}
  It turns out that one of the very first defined CA is a freezing CA. Indeed in \cite{ulam}, as an attempt to study models of cristal growth, the following CA over alphabet ${Q=\{0,1\}}$ is defined: 
  \[F(c)_z =
  \begin{cases}
    0 &\text{ if $c_z=0$ and ${\#_1(c_{|z+V})\neq 1}$}\\
    1 &\text{ else,}
  \end{cases}
  \]
  where $V$ is the von Neumann neighborhood and $\#_1(P)$ denotes the number of occurrences of $1$ in the finite pattern $P$.
\end{example}

\begin{example}
  \label{ex:bootstrap}
  The threshold growth models in 2D \cite{Gravner98} are CA with
  states ${0,1}$ where $0$ becomes $1$ if the number of $1$s in the
  neighborhood is above some threshold, and $1$s stay unchanged
  forever. They were in particular considered as theoretical models of
  bootstrap percolation and a lot of work was dedicated to the
  experimental and rigorous analysis of the phase transitions they
  exhibit (see for instance~\cite{Holroyd}). The fact that these
  examples are monotone (with respect to the order extended to
  configurations) is to be taken into account when studying their computational complexity \cite{GolesMMO17,BeckerMOT18}.
\end{example}

\begin{example}
  \label{ex:lifewithoutdeath}
  By taking any CA on alphabet $Q$ and endowing $Q$ with some order $\leq$, one can define a freezing CA as follows:
  \[F_\leq(c)_z = \min_\leq(c_z, F(c)_z).\]
  There is a priori no relation to expect between $F$ and $F_\leq$,
  but this construction is a source of examples. For instance a freezing
  2D CA with 2 states called ``life without death'' received a lot of
  attention and it was in particular shown that the problem of predicting
  such a CA for a finite amount of  time is $P$-complete \cite{GriMoo96}:
  it is $F_\leq$ where $F$ is the ``game of life'' CA.
\end{example}

\begin{example}
  \label{ex:sir}
  SIR epidemic propagation models \cite{Fuentes}, in their simplest form, are $3$ states freezing CA: a
  \textbf{S}usceptible person can become \textbf{I}nfected, and
  then \textbf{R}ecover and acquire immunity so that it neither
  goes back to susceptible or infected states. A similar CA was also considered to model\footnote{Technically, the CA contains a probabilistic part describing a rate of tree growth. Setting this parameter to $0$ gives a deterministic freezing CA.} forest fire propagation \cite{forestfire}. These examples are in fact multi-state versions of the threshold rules from Example~\ref{ex:bootstrap}.
\end{example}

\begin{example}
  \label{ex:atam}
  Self-assembly tilings \cite{Patitz14,Winslow16}, and more precisely the so called abstract tile assembly model (aTAM), consists in a set of Wang tiles where each edge has a color and a bond strength. They are classically studied as asynchronous non-deterministic systems where tiles are added one by one to the current assembly, starting from a single tile (the seed) and according to the following rule: a tile can be added if the sum of strengths values of edges with a color matching the corresponding neighbor is above some threshold. An important property of this model (no longer true for some of its generalizations like mismatch-free \cite{BeckerM15} or negative glue \cite{DotyPS11,SingleNegative}), is that if a tile can be added at some position in some context, then it can also be added in the context where new tiles are added at empty neighboring positions. Fomally, an aTAM system can be seen as follows:
  \begin{itemize}
  \item $Q = T\cup \{\epsilon\}$ is the alphabet, endowed with the order $\preceq$ such that ${t\preceq\epsilon}$ for all ${t\in T}$;
  \item $s\in T$ is the seed tile;
  \item $U=V\setminus\{(0,0)\}$ where $V$ is the von Neumann neighborhood $V$;
  \item ${R \subseteq Q^U\times Q}$ is the local compatibility rule\footnote{Here we forget about the details of the aTAM model involving bond strength and threshold, we simply observe that it can be represented as a local compatibility relation. Of course not all such relations correspond to an aTAM model.} which is monotone with respect to $\preceq$ (extended to $Q^U$), \textit{i.e.} if ${(m,t)\in R}$ and ${m'\preceq m}$ then ${(m',t)\in R}$, and verifies ${(m,\epsilon)\in R}$ for all ${m\in Q^U}$ (not adding a tile is always possible).
  \end{itemize}
  For two configurations ${c,c'\in Q^{\Z^2}}$ we write ${c\rightarrow c'}$ if $c$ and $c'$ differ at position $z$, ${c_z=\epsilon}$ and ${(c_{|z+U},c'_z)\in R}$. We denote by ${\rightarrow_R^\ast}$ the reflexive and transitive closure of $\rightarrow_R$. An \emph{assembly} is a configuration $c$ with ${c_0\rightarrow_R^\ast c}$ where $c_0$ is the configuration everywhere equal to $\epsilon$ except in position $(0,0)$ where it is $s$. A \emph{terminal assembly} is a configuration $c$ such that ${c_0\rightarrow_R^\ast c}$ and for any ${c'}$ with ${c\rightarrow_R^\ast c'}$ we have ${c=c'}$.
  From such a system we can define a freezing cellular automaton $F_R$ with ordered alphabet ${(Q,\preceq)}$ and neighborhood $V$ as follows: 
  \[F_R(c)_z =
    \begin{cases}
      c_z &\text{ if $c_z\neq\epsilon$},\\
      \min \{t : (c_{|z+U},t)\in R\} &\text{ else.}\\
    \end{cases}
  \]
  In particular, for any finite configuration $c$, ${c\rightarrow_R^\ast F_R(c)}$. We will come back to this simple translation of aTAM systems to freezing CAs in subsection~\ref{sec:freezatam}.
\end{example}

In a freezing CA, a cell can change at most a finite number of time
during its evolutions, precisely at most $n-1$ times for a CA with $n$
states. This property is our second main definition.

\begin{definition}
  A CA $F$ is \emph{k-change} if the number of state change of any cell in any orbit is at most $k$, formally: 
  \[\forall c\in Q^{\zd},\forall z\in \zd : |\{t\in\N : F^{t+1}(c)_z \not= F^t(c)_z\}|\leq k.\]
  A CA $F$ is \emph{bounded-change} if it is $k$-change for some $k$.
\end{definition}

Bounded-change CA  have been studied previously \cite{vollmar81} as
language recognizers. CA with bounded communications have also been considered
 \cite{KutribM10a} (again from the language recognition point of view) and are very close: after a bounded communication a cell do no longer depends on its neighbors and must enter a temporal cycle bounded by the number of states, which means that some power $F^m$ of the CA $F$ is bounded-change --- simply ensure that transitions without communication in $F^m$ are the identity.

A CA $F$ is \emph{nilpotent} if there is ${t>0}$ such that ${F^t}$ is a constant map. Any nilpotent CA is bounded-change.

\begin{example}
  A bounded-change CA is not necessarily freezing. For instance, the
  following one-dimensional nilpotent CA over $\{0,1,2\}$ with neighborhood ${\{0,1\}}$ is not freezing:
  \[f(a,b) =
  \begin{cases}
    1 &\text{ if }b=0\\
    2 &\text{ else.}
  \end{cases}
  \]
  because both ${2\rightarrow 1}$ and ${1\rightarrow 2}$ are possible
  state changes in one cell (but after 2 steps all cells are in state
  $2$).
\end{example}

\newcommand\fzt[1]{\tau(#1)}
If $F$ is bounded-change and $c$ is any configuration then for any cell ${z\in\zd}$ there is some time $\fzt{c,z}$ such that the state of cell $z$ does no longer change after time $\fzt{c,z}$ when starting from $c$. Then for any finite set of cells ${E\subseteq\zd}$, the states of cells in $E$ does no longer change after time ${\max_{z\in E}\fzt{c,z}}$. Said differently, the sequence ${(F^t(c))_{t\in\N}}$ is convergent. This property is our third main definition.

\begin{definition}
  A CA $F$ is \emph{convergent} if, for any configuration $c$, the sequence ${(F^t(c))_{t\in\N}}$ is convergent. In this case, the limit configuration reached is denoted by: 
  \[F^\omega(c) = \lim_{t\rightarrow\infty} F^t(c).\]
  Moreover, given an initial configuration $c$ and a cell ${z\in\zd}$, we call \emph{freezing time} the first time ${\fzt{c,z}}$ after which cell $z$ no longer changes its state, formally: 
  \[\fzt{c,z} = \min \{t\in\N : \forall t'\geq 0, F^{t+t'}(c)_z = F^t(c)_z\}.\]
\end{definition}

The next example gives a construction technique to produce convergent CA that are not bounded-change. The idea is to divide the configuration into working zones and ensure the presence of at most one working head per zone. Each head permanently bounces between the two extremities of its zone and shrinks the zone by one cell at each bounce. With this behavior, each finite working zone gets completely shrinked in finite time and converges to a fixed point. There is no way to check locally that a working zone is finite, but there is no problem of convergence with infinite zones because the head will escape towards infinity after at most one bounce and leaves a fixed point behind. Additionally, the head can do some computation at each pass without breaking the convergence property. In the following example, working zones contain two layers of states from a given CA $F$ interpreted as old state and new state respectively. The head sequentially updates these states to simulate a parallel synchronous application of $F$ at each pass.

\newcommand\convzz[1]{\mathcal{Z}_{#1}}
\begin{example}[Shrinking zone construction]
  \label{ex:szone}
  Let $F$ be any 1D CA on alphabet $Q$ with radius $1$ and local map ${\delta : Q^3\rightarrow Q}$. We define ${\convzz{F}}$ on alphabet ${R = \{b,b_+,e\}\cup Q'}$ with ${Q'=Q\times Q\times\{\leftarrow,\rightarrow,l,r\}}$ and radius $1$ as follows:
  \begin{itemize}
  \item $e$, the \emph{error state}, is a spreading state: any cell with $e$ in its neighborhood turns into state $e$; a configuration $c$ is \emph{valid} if $e$ never appears in its orbit;
  \item  $b$, the blank state, never changes except in presence of the error state; $b_+$ becomes $b$ except in presence of the error state; 
  \item a maximal connected component of cells in state $Q'$ is a \emph{working zone};
  \item in a working zone, patterns of the form ${(x,y,r)(x',y',l)}$, or ${(x',y',l)(x,y,r)}$, or ${(x,y,z)(x',y',z')}$ with ${\{z,z'\}\subseteq\{\leftarrow,\rightarrow\}}$, or ${(x,y,r)(x',y',z)}$ or ${(x,y,z)(x',y',l)}$ with ${z\in \{\leftarrow,\rightarrow\}}$, are forbidden and generate an $e$ state when detected; therefore in a valid configuration and in each working zone there is at most one occurrence of a state of the form ${(x,y,\{\leftarrow,\rightarrow\})}$ called the \emph{head};
  \item a cell without forbidden pattern (from previous item) and without head in its neighborhood doesn't change its state;
  \item the movements and actions of the heads are as follows:
    \begin{itemize}
    \item inside a working zone, the head in state $\leftarrow$ moves left, the head in state $\rightarrow$ moves right; the local map $\delta$ is only applied when the head moves left to right; precisely we have the following transitions:
      \begin{align*}
        (x,y,l) ,\ (x',y',\leftarrow) ,\ (x'',y'',r) &\mapsto (x',y',r)\\
        (x,y,l) ,\ (x',y',l) ,\ (x'',y'',\leftarrow) &\mapsto (x',y',\leftarrow)\\
        (x,y,l) ,\ (x',y',\rightarrow) ,\ (x'',y'',r) &\mapsto (x',y',l)\\
        (x,y,\rightarrow) ,\ (x',y',r) ,\ (x'',y'',r) &\mapsto (\delta(y,x',x''),x',\rightarrow)\\
      \end{align*}
    \item when a boundary of the working zone is reached, the head bounces, changes direction and the working zone get shrinked by one cell; precisely we have the following transitions:
      \begin{align*}
        b,\ (x,y,l),\ (x',y',\leftarrow) &\mapsto (x,y,\leftarrow)\\
        b,\ (x,y,\leftarrow),\ (x',y',r) &\mapsto (x,y,\rightarrow)\\
        b,\ (x,y,\rightarrow),\ (x',y',r) &\mapsto (y,x,l)\\
        b,\ (x,y,l),\ (x',y',\rightarrow) &\mapsto b_+\\
        (x,y,\rightarrow),\ (x',y',r),\ b &\mapsto (x',y',\rightarrow)\\
        (x,y,l),\ (x',y',\rightarrow),\ b &\mapsto (x',y',\leftarrow)\\
        (x,y,l),\ (x',y',\leftarrow),\ b &\mapsto (x',y',r)\\
        (x,y,\leftarrow),\ (x',y',r),\ b &\mapsto b_+\\
      \end{align*}
      \textit{(note the swap between $x$ and $y$ in the third transition above to initialize the sequential application of $\delta$)}
    \item finally the head disappears in a working zone of size $1$, precisely: 
      \[b',(x,y,z),b'' \mapsto (x,y,r)\]
      for any ${b', b''\in \{b,b_+\}}$.
    \end{itemize}
  \end{itemize}
\end{example}

Given two configurations ${c,c'\in Q^\Z}$ and any ${n>0}$ we define the valid configuration ${\lambda_{n,c,c'}\in R^\Z}$ as follows: 
\[\lambda_{n,c,c'}(z) =
\begin{cases}
  b&\text{ if $z<-n$,}\\
  (c_z,c'_z,\rightarrow)&\text{ if $z=-n$},\\
  (c_z,c'_z,r)&\text{ if $-n<z\leq n$},\\
  b&\text{ if $z>n$}.\\
\end{cases}
\]

\begin{lemma}
  For any $F$, the CA ${\convzz{F}}$ is convergent but not bounded-change. Moreover for any configuration ${c}$ of $F$ any ${t>0}$ and any ${n\geq t}$ and any ${z}$ with ${|z|\leq n-t}$ it holds: 
  \[\convzz{F}^{t_n}(\lambda_{n,c,c})_z = \lambda_{n-t,F^t(c),F^{t-1}(c)}\]
  where ${\displaystyle t_n = \sum_{i=n-t+1}^{n}4i+1}$. Moreover, for any ${0\leq t'\leq t_n}$ and for any ${|z|<n}$, ${\convzz{F}^{t'}(\lambda_{n,c,c})_z}$ is either $b$ or $b_+$, or there is some ${0\leq t_1,t_2\leq t}$ such that it is of the form ${(F^{t_1}(c)_z,F^{t_2}(c)_z,x)}$ with ${|t_1-t_2|=1}$.
  \label{lem:szone}
\end{lemma}
\begin{proof}
  First we show that ${\convzz{F}}$ is convergent. To see this let's consider any configuration $c$ and show that ${\bigl(F^t(c)_0\bigr)_{t\in\N}}$ is convergent. We have the following cases:
  \begin{itemize}
  \item if $c$ is not valid then ${F^t(c)_0=e}$ for any large enough $t$;
  \item if $c$ is valid and ${c_0\in\{b,b_+\}}$ then ${F^t(c)_0=b}$ for any $t$;
  \item finally, if $c$ is valid and ${c_0\in Q'}$ then:
    \begin{itemize}
    \item either cell $0$ belongs to an infinite working zone and the head will no longer be in its neighborhood after some finite time its state will no longer change;
    \item or cell $0$ belongs to some finite zone with no head and its state will never change;
    \item or cell $0$ belongs to some finite zone containing a head and after some finite time the zone has shrinked to size $1$ and doesn't change anymore, so that cell $0$ stays in state $b$ or some state ${q\in Q'}$ forever.
    \end{itemize}
  \end{itemize}

  To see that ${\convzz{F}}$ is not bounded-change it is sufficient to check that, for any ${n\in\N}$, cell $0$ changes more than $n$ times in the orbit of ${\lambda_{n,c,c'}}$ whatever ${c,c'\in Q^\Z}$.

  Finally, it is straightforward to check that ${\convzz{F}^{2n+1}(\lambda_{n,c,c'})=d}$ where $d$ is defined by:
  \[d(z) =
  \begin{cases}
    b&\text{ if $z\leq-n$,}\\
    (F(c)_z,c_z,l)&\text{ if $-n<z< n$},\\
    (c_z,c'_z,\leftarrow)&\text{ if $z=n$},\\
    b&\text{ if $z>n$}.\\
  \end{cases}
  \]
  and then ${\convzz{F}^{2n+1+2n}(\lambda_{n,c,c'})=\lambda_{n-1,F(c),c}}$. The formula on ${\convzz{F}^{t_n}(\lambda_{n,c,c})_z}$ follows by induction, and the last point of the lemma follows by definition of $\convzz{F}$ in working zones.
\end{proof}

From the discussion above we have the following strict hierarchy.

\begin{fact}
  $F$ freezing $\begin{matrix}{\Rightarrow}\\{\not\Leftarrow}\end{matrix}$ $F$ bounded-change $\begin{matrix}{\Rightarrow}\\{\not\Leftarrow}\end{matrix}$ $F$ convergent.
\end{fact}

Contrary to the freezing property which is easy to check (Fact~\ref{fact:freezingdecidable}), being bounded-change or convergent are undecidable properties.

\begin{theorem}
  \label{thm:basicundecidable}
  The following holds in any dimension:
  \begin{itemize}
  \item nilpotent CA are recursively inseparable from non-convergent CA;
  \item given $F$ convergent, it is undecidable whether it is bounded-change.
  \end{itemize}
\end{theorem}
\begin{proof}
 In the following we do the proof for dimension 1, but the result follows straightforwardly for any dimension.  In both cases we prove a reduction from the nilpotency problem for 1D CA of radius $1$ with a spreading state. This problem was shown undecidable in \cite{kari92}. Consider any $F$ of radius $1$ on alphabet $Q$ with spreading state ${s\in Q}$ and local map ${\delta: Q^3\rightarrow Q}$, and construct the following CAs:
 \begin{itemize}
 \item $F_1$ on alphabet ${Q\times\{0,1\}}$ with transition map: 
   \[\bigl((q_1,b_1),(q_2,b_2),(q_3,b_3)\bigr)\mapsto
   \begin{cases}
     (s,0)&\text{ if ${s\in\{q_1,q_2,q_3\}}$},\\
     (\delta(q_1,q_2,q_3),1-b_2)&\text{ otherwise}.
   \end{cases}
   \]
   It is straightforward to check that $F_1$ is bounded-change (and more precisely nilpotent) if $F$ is nilpotent and non-convergent else. The first item of the theorem follows;
 \item $F_2$ is ${\convzz{F}}$ with the following modification (using notation of Example~\ref{ex:szone}): whenever a state ${(q,q',x)}$ with ${s\in\{q,q'\}}$ appears in the neighborhood of cell it turns into state $e$. $\convzz{F}$ is always convergent by Lemma~\ref{lem:szone}, and so is $F_2$ because for any initial configuration $c$, either $F_2$ and $\convzz{F}$ produce the same orbit, or $c$ is invalid for $F_2$ (\textit{i.e.}, the orbit converges to ${{}^\omega e^\omega}$). Moreover a straightforward adaptation of Lemma~\ref{lem:szone} shows that: for any configuration ${c}$ of $F$ any ${t>0}$ and any ${n>t}$ it holds: 
  \[\convzz{F}^{t_n}(\lambda_{n,c,c})_0 =
  \begin{cases}
    e &\text{ if ${F^t(c)_0=s}$}\\
    (F^t(c)_0,F^{t-1}(c)_0,h) &\text{ else, with $h\in\{\rightarrow,r\}$}
  \end{cases}
\]
  where $t_n$ is the time constant from Lemma~\ref{lem:szone}. First if $F$ is not nilpotent, then there is a configuration $c$ whose orbit does not contain state $s$. Therefore, for any $n$, ${\lambda_{n,c,c}}$ has an orbit under $F_2$ with at least $n$ changes for cell $0$. This shows that $F_2$ is not bounded-change in this case. If on the contrary $F$ is nilpotent, then there is $t_0$ such that ${F^t(c)_0=s}$ for any ${t\geq t_0}$ and any $c$. Therefore any cell in any working zone of $F_2$ becomes $e$ after $t_0$ passes of the working head. This implies that a cell's state can change at most ${2t_0}$ times while staying in $Q'$ ($t_0$ times head entering/leaving the cell). Since the other possible state changes are ${Q'\rightarrow \{b,b_+\}}$, ${Q'\rightarrow e}$, ${b_+\rightarrow b}$ and ${\{b,b_+\}\rightarrow e}$, we deduce that $F_2$ is ${(2t_0+2)}$-change and the second item of the theorem follows.
 \end{itemize}
\end{proof}

\section{Dynamical Properties}
\label{sec:dynamics}

The convergence property has obviously strong consequences on the dynamics. This can be seen through the trace or column factor which is a well-studied object \cite{kurkabook,guillonphd}.

\newcommand\trace[1]{T_{#1}}
For any finite ${W\subseteq\Z^d}$ and any CA $F$, the trace of $F$ of base $W$ starting from $c$ is the following sequence:
\[\trace{W}(c) = (\restr{F^t(c)}{W})_{t\in\N}.\]

\begin{lemma}
  \label{lem:trace}
  Let ${W\subseteq\zd}$ be any non-empty finite set. A CA $F$ is convergent if and only if $\trace{W}(c)$ is eventually constant for any configuration $c$. Moreover, in that case, there is always a fixed point $c'$ (\textit{i.e.} ${F(c')=c'}$) such that ${\trace{W}(c)}$ and ${\trace{W}(c')}$ are eventually equal.
\end{lemma}

This does not imply that the limit set is made of fixed
points, however we can prove the following proposition which is not
true for CA in general.

\begin{proposition}
  \label{prop:twofix}
  Let $F$ be a convergent CA which is not nilpotent. Then it possesses two
  distinct fixed points. 
\end{proposition}
\begin{proof}
  First, any convergent CA must have a fixed-point $c$ which is uniform: ${c_z =q_0}$ for some ${q_0\in Q}$ and any ${z\in\zd}$. Now suppose that $c$ is the unique fixed-point, then by Lemma~\ref{lem:trace} $F$ is asymptotically-$q_0$-nilpotent, meaning that for any ${z\in\zd}$ and any ${c'\in Q^\zd}$, the trace ${\trace{\{z\}}(c')}$ is eventually constant and equal to $q_0$. By Theorem~2 of \cite{villeasnil12} this implies that $F$ is nilpotent. The proposition follows.
\end{proof}

As a corollary, we get the following dimension-sensitive decidability result.

\begin{theorem}
  The nilpotency problem is:
  \begin{itemize}
  \item decidable (in time polynomial in the size of the transition table) for 1D convergent CA;
  \item undecidable for freezing CA in higher dimension.
  \end{itemize}
\end{theorem}
\begin{proof}
  First, in dimension $2$ and more, the classical proof of
  undecidability of nilpotency (see \cite{kari92}) works without any
  modification for freezing CA: given a Wang tile set, we build a
  freezing CA with a spreading error state, that checks locally if
  the configuration is a valid tiling and produces the error state
  in case of local error detection. This CA is nilpotent if and only
  if the tile set does not tile the plane. The undecidability follows from \cite{berger}.
  
  In dimension $1$, from Proposition~\ref{prop:twofix}, we know that
  nilpotency is equivalent to the property of having a single
  fixed-point. By the classical construction of De Bruijn graph, the existence of at least two fixed-points is easy to
  decide in polynomial time. Precisely, and for the sake of
  completeness, suppose $F$ has radius $r$ and local transition map
  ${f:Q^{2r+1}\rightarrow Q}$ and let ${G=(V,E,\lambda)}$ be the
  edge-labeled digraph with ${V=Q^{2r}}$ and
  \begin{itemize}
  \item ${E = \{\bigl((q_{-r},\ldots,q_{r-1}),(q_{-r+1},\ldots,q_{r})\bigr) : f(q_{-r},\ldots,q_r)=q_0\}}$,
  \item ${\lambda\bigl((q_{-r},\ldots,q_{r-1}),(q_{-r+1},\ldots,q_{r})\bigr) = q_0.}$
  \end{itemize}
  Since $F$ necessarily possesses a uniform fixed point
  ${{}^\omega q^\omega}$, testing if $F$ possesses at least two
  amounts to check whether there is a circuit in $G$ with at least one
  edge not labeled by $q$. This can be done in time polynomial in ${|Q|^r}$.
\end{proof}

Convergence also implies some kind of topological regularity as shown in the following proposition.

\begin{proposition}
  No convergent CA is sensitive to initial conditions, whatever the dimension. 
\end{proposition}
\begin{proof}
  Suppose by contradiction that $F$ on $Q^\zd$ is convergent and sensitive to initial conditions: there is ${N\in\N}$ such that for any ${c\in Q^\zd}$ and any ${p\in\N}$ there exists ${c'\in Q^\zd}$ such that ${\restr{c}{\ball{p}}=\restr{c'}{\ball{p}}}$ and ${\restr{F^t(c)}{\ball{N}}\neq\restr{F^t(c')}{\ball{N}}}$ for some ${t\in\N}$. Consider any configuration ${c_0\in Q^\zd}$ and ${p_1\geq N}$. By sensitivity there is ${c'\in Q^\zd}$ and ${t_1\in\N}$ such that either ${\trace{\ball{N}}(c_0)}$ or ${\trace{\ball{N}}(c')}$ is non-constant on time interval ${[0,t_1]}$. Denote by $c_1$ the one among ${c_0}$ and $c'$ that corresponds to the non-constant trace. Let ${p_2 = p_1 + N + r(t_1+1)}$ where $r$ is the radius of $F$. Applying sensitivity again on $c_1$ and $p_2$ we know there exist $c'$ and $t_2$ such that:
  \begin{itemize}
  \item ${c_1}$ and $c'$ are identical on ${\ball{p_2}}$;
  \item therefore, by choice of $p_2$, ${\trace{\ball{N}}(c_1)}$ and ${\trace{\ball{N}}(c')}$ coincide on interval ${[0,t_1]}$;
  \item ${\trace{\ball{N}}(c_1)}$ and ${\trace{\ball{N}}(c')}$ differ at time ${t_2}$.
  \end{itemize}
  So one of $c_1$ or $c'$, denoted $c_2$, is such that ${\trace{\ball{N}}(c_2)}$ is not constant on interval ${[t_1,t_2]}$. Going on with the same reasoning we construct a converging sequence ${(c_n)_{n\in\N}}$ of configurations such that ${\trace{\ball{N}}(c_n)}$ is not constant on each interval ${[t_i,t_{i+1}]}$ for ${0\leq i<n}$. Taking ${c=\lim_n c_n}$ we get a trace ${\trace{\ball{N}}(c)}$ which is not eventually constant contradicting Lemma~\ref{lem:trace}.
\end{proof}

The convergence behavior is intrinsically irreversible except in the
trivial case of the identity, more precisely there is no surjective convergent CA appart from the identity. It was shown in \cite{GolOlThey15} for freezing cellular automata by an elementary proof. In fact, as pointed to us by V. Salo, it is an immediate consequence of the Poincaré recurrence theorem for all convergent CA: if a CA is not the identity, then their is some local context $u$ in which a cell changes its state, but the Poincaré recurrence theorem implies that in a surjective CA there must be an orbit where $u$ is recurrent which yields infinitely many states change and contradicts convergence. Without going into details about measures and cellular automata which are out of the scope of the present paper (see \cite{Pivato09}), we give a rather self-contained proof of this fact in the following proposition.

\begin{proposition}
  If a convergent CA is surjective, then it is the identity map.
\end{proposition}
\begin{proof}
  Let $F$ be surjective and convergent and let $\mu$ denote the uniform product measure on configurations of $F$. By the balance theorem \cite{maki}, $\mu$ is preserved under $F$: ${\mu(X)=\mu(F^{-1}(X))}$ for any measurable set $X$ (see \cite{CapobiancoGK13} for a more general result). If we suppose that $F$ is not the identity map, then there is a word ${u\in Q^{\ball{n}}}$ such that for all ${x\in[u]}$ it holds ${F(x)_0\neq x_0}$. We claim that there is a configuration $x$ such that ${F^t(x)\in[u]}$ for infinitely many $t$. From this claim we deduce that $F$ is not convergent because the orbit of $x$ is not convergent. The proposition follows. To prove the claim, let us denote ${R_t = \cup_{t'\geq t}F^{-t'}([u])}$ and ${R=\cap_tR_t}$.  By definition $R$ is the set of configurations whose orbit visits $[u]$ infinitely many times and we want to show that ${R\neq\emptyset}$. We actually show that ${\mu(R)>0}$. Since ${R_{t+1}=F^{-1}(R_t)}$ for all ${t\geq 0}$ we have ${\mu(R_t)=\mu(R_0)}$. Moreover ${[u]\subseteq R_0}$ so we deduce that ${\mu([u]\setminus R_t)\leq \mu(R_0\setminus R_t)= 0}$ since ${R_t\subseteq R_0}$. Finally, by Boole's inequality we have ${\mu([u]\setminus R)=0}$ so ${\mu(R)>0}$ since ${\mu([u])>0}$.
\end{proof}

\section{Computational Complexity and Universality}
\label{sec:ccu}

Computational complexity and universality in cellular automata is a very well-studied topic and the general ability for cellular automata to do any Turing computation is now considered as an obvious fact. In this section, we shall show how the three definitions of section~\ref{sec:maindefs} affect this computational power and how it also depends on the dimension.

\subsection{Canonical Problems}

We start by defining two classical problems associated to any CA which will serve as a canonical measurement of the computational complexity of the considered CA. They are very different and complementary.

\newcommand\logred{\leq_{log}}
\newcommand\PRED[1]{\mathrm{PRED}_{#1}}
\newcommand\LIMIT[1]{\mathrm{LIMIT}_{#1}}
\newcommand\TIME[1]{\mathrm{TIME}_{#1}}
\newcommand\CYREACH[1]{\mathrm{CYREACH}_{#1}}
\newcommand\CYREACHOMEGA[1]{\mathrm{CYREACH}^\omega_{#1}}

The first one is about short-term predictability and provides a fine-grained complexity measurement within the class $\PTIME$.

\begin{definition}
  Let $F$ be any CA of radius $r$ and alphabet $Q$.
  The prediction problem $\PRED{F}$ is defined as follows:
  \begin{itemize}
  \item input: ${t>0}$ and ${u\in Q^{\ball{rt}}}$ and ${q\in Q}$
  \item output: decide whether ${F^t(u)=q}$.
  \end{itemize}
\end{definition}

The second one is about long-term reachability and provides a coarse-grained complexity measure that can go (and is expected to go) beyond the decidable. It is inspired by the notion of universality for dynamical symbolic systems from \cite{DelvenneKB06}.

\begin{definition}
  Let $F$ be any CA of radius $r$ and dimension $d$.
  The cylinder reachability  problem $\CYREACH{F}$ is defined as follows:
  \begin{itemize}
  \item input: two bounded configurations $u$ and $v$.
  \item output: decide whether there is ${c\in [u]}$ and ${t\in\N}$ such that ${F^t(c)\in[v]}$.
  \end{itemize}
\end{definition}

For CA in general, hard examples of both problems are well-known. Let us fix a 1D CA $G$ of radius $1$ such that ${\PRED{G}}$ is P-complete, for instance rule 110 \cite{woodsneary06}, and a 1D CA $H$ of radius $1$ such that ${\CYREACH{H}}$ is undecidable. More precisely, we suppose that there is a fixed state ${h\in Q_H}$ of $H$ such that the following sub-problem of ${\CYREACH{H}}$, that we will denote ${\CYREACH{H}^\ast}$, is also undecidable: given a bounded configuration ${u}$, decide whether there is ${t\in\N}$ and ${c\in[u]}$ such that ${F^t(c)\in[h]}$. Such an $H$ can be obtained by adapting the example of \cite[section 6.1]{DelvenneKB06}. When restricting to freezing CA, and provided the dimension is at least 2, one can find hard examples for both problems. Note that ``life without death'' is known to have a P-complete ${\PRED{}}$ problem \cite{GriMoo96} but the hardness of ${\CYREACH{}}$ is not clear.

\begin{proposition}
  \label{prop:2Dhardfreezing}
  There exists a 2D freezing CA $F$ such that ${\PRED{F}}$ is P-complete and ${\CYREACH{F}}$ is undecidable.
\end{proposition}
\begin{proof}
  Any 1D CA $F$ with states $Q$ and neighborhood $V$ can be simulated by a 2D freezing CA $F'$ with
  states ${Q\cup\{\ast\}}$ as follow. Let ${V'=\{(v,-1):v\in V\}}$. A cell in a state from $Q$ never
  changes. A cell in state $\ast$ looks at cells in its $V'$
  neighborhood: if they are all in a state from $Q$ then it updates to
  the state given by applying $F$ on them, otherwise it stays
  in $\ast$. Starting from a all-$\ast$ configuration except on one
  horizontal line where it is in a $Q$-configuration $c_0$, this 2D
  freezing CA will compute step by step the space-time diagram of $F$
  on configuration $c_0$. Applying this construction to $G$, we obtain a 2D freezing CA $G'$ such that ${\PRED{G}}$ reduces to ${\PRED{G'}}$ as follows: given ${t>0}$ and ${u\in Q_G^{\ball{t}}}$ (of dimension 1), we compute ${u'\in Q_{G'}^{\ball{t}}}$ (of dimension 2) defined by 
  \[u'_z =
  \begin{cases}
    u_x &\text{ when ${z=(x,-t)}$,}\\
      \ast&\text{ else.}
  \end{cases}
  \]
  By definition of $G'$ it holds: ${(G')^t(u') = G^t(u)}$.

  We now apply essentially the same construction to $H$ to obtain $H'$, but with two additional tweaks that allow a reduction from ${\CYREACH{H}^\ast}$ to ${\CYREACH{H'}}$: when $H'$ simulates $H$ by progressively building successive configuration of the orbit of $H$ as successive rows in direction ${(0,1)}$, it also propagates back any state $h$ appearing in the simulation in direction ${(0,-1)}$; moreover, a special state $b$ is added whose role is to appear and spread in any ill-formed configuration. With these tweaks, to any cylinder $[u]$ for $H$ we can associate a cylinder $[u']$ for $H'$ such that $[h]$ is reachable from $[u]$ in $H$ if and only if $[h]$ is reachable from $[u']$ is $H'$.
  To do so, the construction above is applied with the following modifications:
  \begin{itemize}
  \item $H'$ has neighborhood $\ball{1}$ which is large enough to have ${V'\subseteq \ball{1}}$ since $H$ has radius $1$;
  \item $H'$ as an additional state $b$ that never changes whatever the context, and that propagates downwards and to the left and to the right;
  \item if a cell in a state from $Q$ has a left or right neighbor not in $Q$, it becomes $b$;
  \item if a cell $z$ in state $\ast$ has a neighbor which is not in $z+ V'$ and not in state $\ast$ then it becomes $b$;
  \item if none of the above cases holds and cell $z$ is in a state from $Q$ and cell ${z+(0,1)}$ is in state $h$ then cell $z$ turns into state $h$.
  \end{itemize}
  Then ${\CYREACH{H}^\ast}$ reduces to ${\CYREACH{H'}}$ as follows: given a bounded configuration $u$ of radius $n$, we compute ${u'\in Q_{H'}^{\ball{n}}}$ defined by
  \[u'_z =
  \begin{cases}
    u_x &\text{ when ${z=(x,0)}$,}\\
    b &\text{ if ${z=(x,y)}$ with ${y<0}$,}\\
    \ast&\text{ else.}
  \end{cases}
  \]
  Then by definition of $H'$ we have that ${H^t(c)\in[h]}$ for some ${t\in\N}$ and ${c\in[u]}$ if and only ${(H')^{t'}(c')\in[h]}$ for ${t'\in\N}$ and ${c'\in[u']}$. Indeed, if ${H^t(c)_0=h}$ then ${(H')^t(c')_{(0,t)}=h}$ and therefore ${(H')^{2t}(c')_{(0,0)}=h}$ where $c'$ is obtained from $c$ as $u'$ is from $u$. For the other direction, ${(H')^{t'}(c')_{(0,0)}=h}$ implies that there is some $t\leq t'$ with ${(H')^{t}(c')_{(0,t)}=h}$ and it can only be possible if ${c'_{(x,0)}\in Q}$ for all ${|x|\leq t}$ and ${c'_{(x,y)}=\ast}$ for all $(x,y)\in\ball{t}$ and ${y>0}$ (because no $b$ can be produced close to $(0,0)$ at first step). We deduce that there is a configuration ${c\in[u]}$ with ${c_x = c'_{(x,0)}}$ for all ${|x|\leq t}$ for which it holds ${H^t(c)_0=h}$. The reduction from ${\CYREACH{H}^\ast}$ to ${\CYREACH{H'}}$ follows.

  To conclude, the product cellular automaton ${G'\times H'}$ that acts independently as $G'$ and $H'$ on both components of the product configuration set with set of states $Q_{G'}\times Q_{H'}$, fulfills the conditions of the proposition because since both $G'$ and $H'$ have a quiescent state, then ${\PRED{G'}}$ reduces to ${\PRED{G'\times H'}}$ and ${\CYREACH{H'}}$ reduces to ${\CYREACH{G'\times H'}}$.
\end{proof}

\begin{proposition}
  \label{prop:hard1Dconvergent}
  There exists a 1D convergent CA $F$ such that ${\PRED{F}}$ is P-complete and ${\CYREACH{F}}$ is undecidable.
\end{proposition}
\begin{proof}
  This proof relies on construction from Example~\ref{ex:szone} which we apply on automata $G$ and $H$ mentioned above in this section. We claim that ${\convzz{G}\times\convzz{H}}$ has the desired property. Indeed, on one hand ${\PRED{G}}$ reduces to ${\PRED{\convzz{G}}}$ as follows: given ${t>0}$ and ${u\in Q_G^{\ball{t}}}$, we compute ${t_n}$ from Lemma~\ref{lem:szone} with ${n=t+1}$ and ${u'\in Q_{\convzz{G}}^{\ball{t_n}}}$ such that ${u' = \restr{\lambda_{t_n,c,c}}{\ball{t_n}}}$ where ${\restr{c}{\ball{t}}=u}$ and $c$ is constant outside $\ball{t}$. Lemma~\ref{lem:szone} shows that the first component of state ${\convzz{G}^{t_n}(u')}$ is precisely ${G^t(u)}$. The reduction follows. On the other hand ${\CYREACH{H}^\ast}$ truth-table reduces to ${\CYREACH{\convzz{H}}}$ as follows: there is $t$ and ${c\in[u]}$ with ${H^t(c)\in[h]}$ if and only if there is $t$ and ${c'\in[u']}$ with ${\convzz{H}^t(c')\in[v']}$ where 
  \begin{itemize}
  \item ${u'_z = \bigl(u_z,u_z,r\bigr)}$ and $u$ and $u'$ have same
    domain;
  \item ${v'}$ is one of the (finite number of) bounded configurations of the form ${v'_z = \bigl(h,x,y\bigr)}$.
  \end{itemize}
  Indeed, if ${H^t(c)\in[h]}$ for some ${c\in[u]}$ then Lemma~\ref{lem:szone} shows that ${\convzz{H}^{t_n}(\lambda_{n,c,c})\in[v']}$ for ${n>t}$ and some $v'$ as above and we have ${\lambda_{n,c,c}\in[u']}$. Conversely if ${\convzz{H}^t(c')\in[v']}$ for some $v'$ as above and ${c'\in[u']}$, then necessarily $c'$ contains a working zone that (possibly) extends $u'$ and survives at least $t$ steps. Then, either this zone contains no head and then states don't change inside the zone which implies ${c'\in[v']}$ and therefore ${[u]\cap[h]\neq\emptyset}$, or the orbit of $c'$ coincides during $t$ steps with that of some ${\lambda_{n,c,c}}$ with ${c\in[u]}$ and Lemma~\ref{lem:szone} implies that ${H^{t'}(c)\in[h]}$. To conclude that ${\convzz{G}\times\convzz{H}}$ has the desired property, it is enough to remark that both $\convzz{G}$ and ${\convzz{H}}$ have a quiescent state so that ${\PRED{\convzz{G}}}$ reduces to ${\PRED{\convzz{G}\times\convzz{H}}}$ and ${\CYREACH{\convzz{H}}}$ reduces to ${\CYREACH{\convzz{G}\times\convzz{H}}}$.
\end{proof}

The following result was essentially present in~\cite{vollmar81} using the formalism of bounded-change CA seen as language recognizers and therefore focusing on the complexity of recognizable languages rather than prediction problems.

\begin{proposition}
  \label{prop:bclogspace}
  Let $F$ be any 1D bounded-change CA. Then $\PRED{F}\in\NL$. Moreover, if in addition the neighborhood $V$ of $F$ is ``one-way'', \textit{i.e.\/} if ${V\subseteq \N}$ or ${V\subseteq -\N}$ then $\PRED{F}\in\LOG$.
\end{proposition}
\begin{proof}
  Consider a $k$-change CA $F$ of
  radius $r$ and alphabet $Q$. Given some ${u\in Q^{\ball{rt}}}$, the following non-deterministic algorithm allows to compute ${F^t(u)}$:
  \begin{enumerate}
  \item for each $i$ with ${-rt\leq i\leq rt}$ do:
    \begin{enumerate}
    \item guess some column vector of states of height ${rt-|i|+1}$ whose first state is ${u_i}$: $C_i\in Q^{rt-|i|+1}$ with ${C_i(0)=u_i}$;
    \item if ${-rt+r\leq i\leq rt-r}$ then check that $C_i$ is compatible with ${C_{i-r},\ldots,C_{i+r}}$ \textit{i.e.} 
      \begin{itemize}
      \item for each $\tau$ with ${0< \tau\leq |i|/r}$ check that \[C_i(\tau) = f(C_{i-r}(\tau-1),\ldots,C_{i+r}(\tau-1))\] where $f$ is the local map of $F$.
      \end{itemize}
    \end{enumerate}
  \item return $C_0(t)$
  \end{enumerate}
  
  The key observation is that in a $k$-change CA a column of states of height at most $t$ appearing in a valid space-time diagram can be represented in space ${O(\log(t))}$ as a list of at most $k$ pairs (state,duration) describing the sequence of state changes and the duration of the intervals between them. Therefore, the above algorithm can be implemented in non-deterministic logspace by keeping in memory only columns $C_{i-r}$ to $C_{i+r}$. The compatibility test (loop on $\tau$) can be done directly on the compact representation of columns since $C_i(\tau)$ can be retrieved by a bounded number of comparisons and sums of small integers. This shows ${\PRED{F}\in\NL}$. 

  If we suppose that ${V\subseteq -\N}$, then the compatibility test only involves column ${C_{i-r}}$ to $C_i$. Moreover, there is a unique possible column ${C_i}$ compatible with the given columns $C_{i-r},\ldots, C_{i-1}$ which can be computed by ${C_i(\tau) = f(C_{i-r}(\tau-1),\ldots,C_{i}(\tau-1))}$. Again this can be realized using compact representation using ${O(\log(t))}$ space. We deduce in this case that ${\PRED{F}\in\LOG}$.
\end{proof}

Under the hypothesis ${\PTIME\neq \NL}$, bounded-change CA have simpler prediction problems than convergent CA. The goal of the next two subsections is to show that the same happens when considering the communication complexity of prediction problems, but that freezing CA can have undecidable cylinder reachability. Hence the difference between bounded-change and convergent CA is in the short-term prediction complexity rather than in the long-term reachability complexity.

\subsection{Communication complexity}

Communication complexity was introduced by Yao~\cite{yao79} to study
distributed computing.
We first briefly recall the classical definition of communication
complexity for any function (see~\cite{kushilevitz97} for general
reference), and then apply it to the prediction problem of cellular
automata as it was done in \cite{ccca}. The goal of this section is
to show that bounded-change CA have a lower communication complexity than CA
in general.

Consider a function ${f:A\times B\rightarrow C}$ where $A$, $B$ and
$C$ are finite sets. The communication complexity of $f$ is the amount
of information (number of bits) that need to be exchanged in the worst
case between two parties (say Alice and Bob) in order to decide the
value $f(x,y)$ when one of the parties knows only $x$ and the other
only $y$ (both know $f$ and they have unlimited computing power).

\newcommand\CC[1]{\textsc{cc}({#1})}

More precisely, a \emph{communication protocol} $\pi$ is a finite rooted binary tree where each leaf holds a value from $C$ and each internal node $n$ holds either a map ${\alpha_n:A\rightarrow \{0,1\}}$ (Alice speaks) or a map ${\beta_n:B\rightarrow \{0,1\}}$ (Bob speaks). The run of $\pi$ on input ${(a,b)}$ is the path ${\rho(a,b)}$ starting from the root and defined by the following descent in the tree until reaching a leaf: when at node $n$, go to left child if ${\alpha_n(a)=0}$ (resp. ${\beta_n(b)=0}$) and to right child otherwise. The output ${\pi(a,b)\in C}$ of the protocol $\pi$ on input ${(a,b)}$ is the value held by the leaf reached by ${\rho(a,b)}$. The cost $c_\pi(a,b)$ of $\pi$ on input ${(a,b)}$ is the length of ${\rho(a,b)}$.  A protocol \emph{solves $f$} if ${f(a,b)=\pi(a,b)}$ on all inputs ${(a,b)\in A\times B}$. Finally, the \emph{communication complexity} of $f$ is 
\[\CC{f}=\min_{\pi\text{ solves $f$}}\max_{a,b}c_\pi(a,b).\]


For any CA $F$ we consider the communication complexity of the
prediction problem by dividing the space into two roughly equal parts through an
hyperplane (we extend the definition of~\cite{ccca} to any dimension). Precisely, given any ${n\in\N}$ we define $A_n$ and $B_n$ as follows:
\begin{align*}
  A_n &= \ball{n}\cap\{z\in\Z^d : \pi_1(z) \leq 0\},\\
  B_n &= \ball{n}\cap\{z\in\Z^d : \pi_1(z) > 0\},
\end{align*}
where ${\pi_1:\Z^d\rightarrow\Z}$ denotes the projection on the first coordinate. The proof idea in Theorem~\ref{thm:cc} below works for other choices of hyperplane, but we prefer to keep notations simple.

For any alphabet $Q$ let ${(u,v)\in Q^{A_n}\times Q^{B_n}\mapsto u|v \in Q^{\ball{n}}}$ denote the concatenation map defined as the inverse map of ${w\in Q^{\ball{n}}\mapsto \bigl(\restr{w}{A_n},\restr{w}{B_n}\bigr)}$. Given a CA $F$ of radius $r$, we then define for each $n$ the communication problem ${\mathcal{P}_{F,n}}$ associated to $\PRED{F}$ as follows: 
\[(u,v)\in Q^{A_{rn}}\times Q^{B_{rn}}\mapsto F^n(u|v) \in Q\]
and we consider
\newcommand\COM[1]{CC_{#1}}
the communication complexity of prediction of $F$ through the following map: 
\[\COM{F}: n\in\N\mapsto\CC{\mathcal{P}_{F,n}}.\]

Clearly ${\COM{F}\in O(n^d)}$ for any $F$ of dimension $d$ since it is always possible to apply the trivial protocol where Alice sends all its input to Bob and Bob then knows all the information to answer. Let us first recall that general CAs can reach maximal communication complexity (this is a straightforward extension of \cite{ccca} which is written in the one-dimensional setting). The intuition is that a $d$-dimension CA in time $n$ can test equality between ${\Omega(n^d)}$ pairs of bits, one in the Alice region, one in the Bob region, and have the conjunction of all these tests readable at a given position. 

\begin{proposition}
  \label{prop:highcc}
  For any $d$, there is a CA $F$ of dimension $d$ with ${\COM{F}\in\Omega(n^d)}$ and it can be chosen of radius 1.
\end{proposition}
\begin{proof}
  Let $e_1=(1,0,\ldots,0)$ and define $F$ of neighborhood ${B(1)}$ on alphabet ${Q=\{\leftarrow,\rightarrow,T\}\times\{0,1\}}$ as follows:
  \[F(c)_z =
  \begin{cases}
    c_{z-e_1}&\text{ if ${c_z\in\{\rightarrow\}\times\{0,1\}}$},\\
    c_{z+e_1}&\text{ if ${c_z\in\{\leftarrow\}\times\{0,1\}}$},\\
    (T,0)&\text{ else if ${\exists z'\in B(1)}$ with ${c_{z+z'}=(T,0)}$ or ${\pi_2(c_{z-e_1})\neq\pi_2(c_{z+e_1})}$},\\
    (T,1)&\text{ in any other case.}
  \end{cases}\]
  Consider configurations $c$ satisfying the following property :
  \begin{itemize}
  \item ${c_z \in (\rightarrow,\{0,1\})}$ if ${\pi_1(z)<0}$,
  \item ${c_z \in (\leftarrow,\{0,1\})}$ if ${\pi_1(z)>0}$,
  \item ${c_z = (T,1)}$ else.
  \end{itemize}
  We call such configurations \emph{correct}. For any
  ${z=(i_1,\ldots,i_d)\in\Z^d}$, consider
  ${\overline{z}=(-i_1,i_2,\ldots,i_d)}$. Let $X_n$ be the set of correct configurations such that ${\pi_2(c_z)=\pi_2(c_{\overline{z}})}$ for all ${z\not\in\ball{n}}$. It is straightforward to show that for ${c\in X_{n/2}}$ it holds that ${\pi_2(F^{n}(c)_0)=1}$ if and only if ${\pi_2(c_z)=\pi_2(c_{\overline{z}})}$ for all ${z\in\ball{n/2}}$. If we now consider the problem ${\mathcal{P}_{F,n}}$, we deduce a fooling set of size ${2^{n^d/2}}$, \textit{i.e.} a set of input ${(u_i,v_i)}$ such that for all ${i\neq j}$ it holds ${F^n(u_i,v_j)\neq F^n(u_i,v_i)}$. This implies ${\COM{F}\in\Omega(n^d)}$ (see \cite{kushilevitz97}).
\end{proof}

Bounded-change CA are intrinsically limited in communication complexity compared to general CA as shown by the following theorem.

\begin{theorem}
  \label{thm:cc}
  For any bound-change CA $F$ of dimension $d$, we have ${\COM{F}\in O\bigl(n^{d-1}\log(n)\bigr)}$.
\end{theorem}
\begin{proof}
  Let $r$ be the radius of $F$ and $k$ be such that $F$ is $k$-change.
  Fix some ${n\in\N}$. We distinguish two zones of the space that play
  an important role: 
  \begin{align*}
    Z_A &= A_{rn} \cap \{z\in\Z^d : -r+1\leq\pi_1(z)\leq 0\},\\
    Z_B &= B_{rn} \cap \{z\in\Z^d : 1\leq\pi_1(z)\leq r\}.\\
  \end{align*}
  These zones are such that Alice (resp. Bob) only needs to know the
  configuration on ${A_{rn}\cup Z_B}$ (resp. ${B_{rn}\cup Z_A}$) to know the configuration after one step of $F$ on $A_{r(n-1)}$ (resp. $B_{r(n-1)}$).
  The protocol is as follows: first, as an initialization step, Alice
  and Bob exchange the content of the $Z_A$ and $Z_B$ and initialize
  their time counter at $0$. Then, they repeat as much as necessary
  the following loop:
  \begin{enumerate}
  \item Alice (resp. Bob) separately compute the evolution step by
    step on $A_{rn}$ (resp. $B_{rn}$) \textbf{supposing that nothing changes}
    in $Z_B$ (resp. $Z_A$) and \textbf{until a state change is
      observed} in $Z_A$ (resp. $Z_B$);
  \item Alice (resp. Bob) send the time counter at which her change was observed; let $t_m$ be the minimum of the time value between Alice and Bob;
  \item if Alice (resp. Bob) has a time counter equal to $t_m$ then she
    sends a ``diff report'' which contains the list of cell changes
    that occurred inside $Z_A$ (resp. $Z_B$) between the previous time counter value and $t_m$;
  \item Alice (and symmetrically for Bob) possibly receives the other ``diff report'' and
    updates its knowledge as follows:
    \begin{itemize}
    \item updates her knowledge of $Z_B$ according to the diff
      report (if received);
    \item if its time counter value is strictly larger than $t_m$,
      then revert its knowledge of her half-space to what it was at
      time $t_m$ (which she can do without any further information
      from Bob).
    \item sets its time counter as $t_m$;
    \end{itemize}
  \end{enumerate}
  This protocol allows Alice and Bob to correctly compute the
  evolution of $F$ during $n$ steps (without even supposing that
  it is bounded-change) and thus solves ${\mathcal{P}_{F,n}}$.
  Indeed, the fact that they retrospectively jump back in time to
  the step of the first change in ${Z_A\cup Z_B}$ zones ensures that the
  hypothesis of no change in ${Z_A\cup Z_B}$ is correct for all the
  retained computation steps.
  
  Let us now upper bound the cost of this protocol:
  \begin{itemize}
  \item the initialization costs $O(n^{d-1})$ (the sizes of $Z_A$ and $Z_B$ are $rn^{d-1}$);
  \item the total communication for time counters is ${O(n\log(n)}$ since the main loop of the protocol is executed at most $n$ times;
  \item each diff report with $i$ cell changes costs at most ${i(\log(|Q|) + (d+1)\log(n))}$:
    $\log(n)$ for the time step, $d\log(n)$ for
    communicating the position of each cell, $\log(|Q|)$ to describe the
    new state;
  \end{itemize}
  The protocol is such that each cell change of each diff report sent corresponds to a real state change in ${Z_A\cup Z_B}$ during the first $n$ steps of the evolution of $F$. Since $F$ is $k$-change, the total number of changes in
  ${Z_A\cup Z_B}$ is less than ${2rkn^{d-1}}$. Therefore the total cost of the protocol is at most ${O\bigl(n^{d-1}\log(n)\bigr)}$.
\end{proof}

Convergent CA in general require higher communication complexity as shown in the following proposition.

\begin{proposition}
  \label{prop:highccconvergent}
  There exists a 1D convergent CA $F$ with ${\COM{F}\in\Omega(\sqrt{n})}$.
\end{proposition}
\begin{proof}
  Let $G$ be any 1D CA of radius 1 such that ${\COM{G}\in\Omega(n)}$ (it exists by Proposition~\ref{prop:highcc}). We claim that ${F=\convzz{G}}$ is such that ${\COM{F}\in\Omega(\sqrt{n})}$. To see this consider for any $n$ the communication problem ${\mathcal{P}_{G,n}}$ associated to $G$ on some input ${(u,v)\in Q_G^{A_{n}}\times Q_G^{B_{n}}}$ and let ${(u',v')=(\restr{c'}{A_{\tau_n}},\restr{c'}{B_{\tau_n}})}$ where ${c'=\lambda_{n,c,c}}$ and ${\tau_n}$ is the constant from Lemma~\ref{lem:szone} that guarantees \[\convzz{G}^{\tau_n}(\lambda_{n,c,c})_0 = (G^n(c)_0,G^{n-1}(c)_0,r).\]
This transformation from ${(u,v)}$ to ${(u',v')}$ is a reduction of the communication problem ${\mathcal{P}_{G,n}}$ to ${\mathcal{P}_{\convzz{G},\tau_n}}$ since $u'$ (resp. $v'$) can be computed from $u$ only (resp. $v$ only), so we have ${\CC{\mathcal{P}_{G,n}}\leq\CC{\mathcal{P}_{\convzz{G},\tau_n}}}$. Since ${\tau_n\leq kn^2}$ for some ${k>0}$ we deduce that ${\COM{\convzz{G}}(n) \geq \COM{G}\bigl(\sqrt{n/k}\bigr)}$. The proposition follows from the hypothesis on $\COM{G}$.
\end{proof}

\subsection{Minsky machines simulation by 1D freezing CA}
\label{sec:minsky}

In this section we show that any $k$-counter Minsky
machine~\cite{Minsky} can be simulated by a 1D freezing CA where the
evolution of both the state and counter values are encoded by signals
moving through the space-time diagram, one transition per space unit. The instant configuration (state + counters) of the counter machine at time $t$ is represented by the temporal column (trace) of cell $t$ of the cellular automaton, where the value of counters is encoded in  unary as the (temporal) distance between occurrences of two special states of the CA.

More precisely, the control of the machine is represented as a signal
of equation $0\leqslant y-2x\leqslant 1$ where both cells in column
$x$ carry the state of the counter machine and an emptiness flag for
each counter at time $x$. A counter can easily be tested for equality
to zero. Incrementation corresponds to a local deceleration of the
signal and decrementation to a local acceleration. Incrementation,
decrementation or no-operations decisions are carried from the control
to the counter as a simple vertical signal.

This encoding scheme, already presented in \cite{GolOlThey15} and detailed more recently in \cite{CartonGR18}, can convince the reader that 1D freezing CA are 'Turing-universal' devices. However, our goal is to prove undecidability results about problem $\CYREACH{}$ and the maximal number of state change per cell in a freezing CA, so we present a construction with additional technical details useful for our goals.

\begin{definition}
  A (deterministic) $k$-\emph{counter Minsky machine} is a $4$-tuple ${M=(Q_M,q_0,h,\tau)}$ where ${q_0,h\in Q_M}$ are the initial and halting states and ${\tau : Q_M\times\{0,1\}^k\rightarrow Q_M\times\{-1,0,1\}^k}$ is its transition map, which verifies ${\tau(h,\cdot)=(h,(0,\ldots,0))}$. A \emph{configuration} of $M$ is an element of ${Q_M\times\N^k}$. $M$ changes any configuration ${c=(q,(\chi_i)_{1\leq i\leq k})}$ in one step into configuration ${M(c)=(q',(\max(0,\chi_i+\delta_i))_{1\leq i\leq k})}$ where ${(q',(\delta_i)_{1\leq i\leq k}) = \tau (q,(\min(1,\chi_i))_{1\leq i\leq k})}$. $M$ halts on input ${(\chi_i)_{1\leq i\leq k}\in\N^k}$ if there is a time $t$ such that ${M^t(q_0,(\chi_i)_{1\leq i\leq k})\in (h,\N^k)}$.
\end{definition}

Minsky machines might be slow, but they have the features of an acceptable computational model and in particular have an undecidable halting problem.

\begin{theorem}
  \label{theo:minsky}
  \cite{Minsky}
  There is a $2$-counter machine $M_u$ such that it is undecidable to know whether $M_u$ halts on a given input. It is also undecidable whether some given $2$-counter machine halts on the empty input (all counters equal to zero). 
\end{theorem}

We now detail our construction that transforms any $k$-counter machine ${M=(Q_M,q_0,h,\tau)}$ into a 1D freezing CA $F_M$ that simulates it. Its neighborhood is ${V=\{-1,0,1\}}$ and its state set is \[Q=\{b,w\}\cup Q_c^k\cup Q_s\cup I\] where:
\begin{itemize}
\item $b$ is a ``blank'' symbol used to fill a configuration and leave room for a computation to expand spatially,
\item $w$ is a ``wall'' state to protect computation zones and avoid back propagation of wrong halting information,
\item ${Q_c = \{1,\#_{-1},\#_0,\#_1\}\times \{-1,0,1\}}$ is the
  \emph{counter alphabet} used to encode the value of one counter and
  the elementary operation to be applied on it,
\item ${Q_s = Q_M\times\{-1,0,1\}^k\cup Q_M}$ is the \emph{control
    alphabet} used to encode the state of M (an element of ${Q_M}$)
  and (possibly) an elementary operation per counter,
\item ${I=\{i_0,\ldots,i_{K}\}}$ is a ``large enough'' ordered set serving as a bounded countdown before initializing a computation of $M$ on an empty input in such a way that it maximizes the number of state changes in a cell; in practice we fix ${K=3k+3}$.
\end{itemize}

Before giving the transition rule of $F$, let us give the (pre-)order $\preceq$ on $Q$ ensuring that it is a freezing CA together with some useful intuitions about the behavior of $F$:
\begin{itemize}
\item ${h\preceq q_c}$ for any ${q_c\in Q_c}$ (\emph{$h$ is used to propagate the halting information of an halting computation of $M$ back to the cell where it started}),
\item ${w\preceq q}$ for any ${q\in Q}$ (\emph{$w$ can be triggered from any state}),
\item ${(c_i,\delta_i)_{1\leq i\leq k}\preceq (c_i',\delta_i)_{1\leq i\leq k}}$ if ${c_i\prec c'_i}$ for all i where ${\#_1\prec\#_0\prec\#_{-1}\prec1}$ (\emph{a counter is represented as a sequence of 1s terminated by the sequence ${\#_{-1}\#_0\#_1}$, used to adjust the offset of the next counter value on the next column, and carrying an invariable operation}),
\item ${q_c\preceq q_s'\preceq q_s\preceq b}$ for any ${q_c\in Q_c,q_s'\in Q_M\times\{-1,0,1\}^k,q_s\in Q_M}$ (\emph{the ``normal sequence'' of state changes is: blank, then control state, then control state plus instruction, then counter}),
\item ${(\#_{-1},0)^k\preceq i_K\preceq\cdots\preceq i_0}$ and ${(q_0,0^k)\preceq i_K}$ (\emph{the special sequences with  $K$ state changes to initialize a computation on an empty input}).
\end{itemize}

\newcommand\cll[4]{\draw[fill=#3] (#1,#2) rectangle ++(1,1) node[pos=.5] {#4};}
\newcommand\ceh[2]{\cll{#1}{#2}{green!50!white}{$h$}}
\newcommand\cec[3]{\cll{#1}{#2}{blue!50!white}{#3}}
\newcommand\ces[3]{\cll{#1}{#2}{red!50!white}{#3}}
\newcommand\ceb[2]{\cll{#1}{#2}{white}{$b$}}
\newcommand\cew[2]{\cll{#1}{#2}{gray}{$w$}}
\newcommand\cei[3]{\cll{#1}{#2}{yellow!50!white}{#3}}
\newcommand\cedh[2]{\cll{#1}{#2}{white}{$\cdots$}}
\newcommand\cedv[2]{\cll{#1}{#2}{white}{$\vdots$}} 

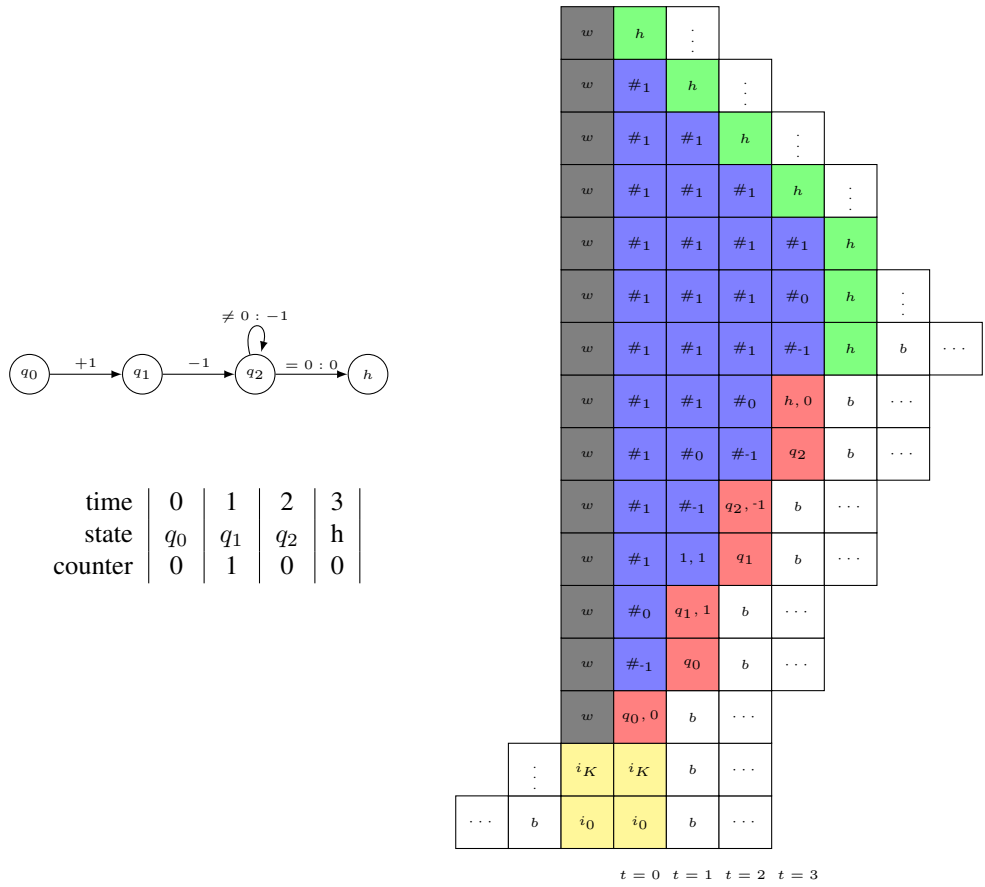
\begin{figure} 
  \label{fig:deffm}
  \begin{center}
    \begin{minipage}[c]{0.45\linewidth}
      \begin{center}
        \begin{tikzpicture}[scale=1]\tiny
          \tikzstyle{config} = [draw,circle,minimum size=15,fill=white]
          \node[config] (qz) at (0,0) {$q_0$};
          \node[config] (qu) at (1.5,0) {$q_1$};
          \node[config] (qd) at (3,0) {$q_2$};
          \node[config] (qh) at (4.5,0) {$h$};
          \draw[->] (qz) edge  node[above]{$+1$}  (qu);
          \draw[->] (qu) edge  node[above]{$-1$}  (qd);
          \draw[->] (qd) edge[loop above]  node[above]{$\neq 0 : -1$}  (qd);
          \draw[->] (qd) edge  node[above]{$=0 : 0$}  (qh);
        \end{tikzpicture}
        \end{center}\vskip 1cm
        \begin{center}
        \begin{tabular}{r|c|c|c|c|}
          time &0&1&2&3\\
          state & $q_0$&$q_1$&$q_2$ &h\\
          counter &0&1&0&0
        \end{tabular}
      \end{center}
    \end{minipage}
    \begin{minipage}[c]{0.5\linewidth}
      \begin{tikzpicture}[scale=.7]\tiny
        \def\-{\raisebox{.75pt}{-}}
        \draw (2.5,-0.5) node {$t=0$};
        \draw (3.5,-0.5) node {$t=1$};
        \draw (4.5,-0.5) node {$t=2$};
        \draw (5.5,-0.5) node {$t=3$};
        \cedh{-1}{0}\ceb{0}{0}\cei{1}{0}{$i_0$}\cei{2}{0}{$i_0$}\ceb{3}{0}\cedh{4}{0} 
        \cedv{0}{1}\cei{1}{1}{$i_K$}\cei{2}{1}{$i_K$}\ceb{3}{1}\cedh{4}{1} 
        \cew{1}{2}\ces{2}{2}{$q_0,0$}\ceb{3}{2}\cedh{4}{2} 
        \cew{1}{3}\cec{2}{3}{$\#_{\-1}$}\ces{3}{3}{$q_0$}\ceb{4}{3}\cedh{5}{3} 
        \cew{1}{4}\cec{2}{4}{$\#_{0}$}\ces{3}{4}{$q_1,1$}\ceb{4}{4}\cedh{5}{4} 
        \cew{1}{5}\cec{2}{5}{$\#_{1}$}\cec{3}{5}{$1,1$}\ces{4}{5}{$q_1$}\ceb{5}{5}\cedh{6}{5} 
        \cew{1}{6}\cec{2}{6}{$\#_{1}$}\cec{3}{6}{$\#_{\-1}$}\ces{4}{6}{$q_2,\-1$}\ceb{5}{6}\cedh{6}{6} 
        \cew{1}{7}\cec{2}{7}{$\#_{1}$}\cec{3}{7}{$\#_{0}$}\cec{4}{7}{$\#_{\-1}$}\ces{5}{7}{$q_2$}\ceb{6}{7}\cedh{7}{7} 
        \cew{1}{8}\cec{2}{8}{$\#_{1}$}\cec{3}{8}{$\#_{1}$}\cec{4}{8}{$\#_{0}$}\ces{5}{8}{$h,0$}\ceb{6}{8}\cedh{7}{8} 
        \cew{1}{9}\cec{2}{9}{$\#_{1}$}\cec{3}{9}{$\#_{1}$}\cec{4}{9}{$\#_{1}$}\cec{5}{9}{$\#_{\-1}$}\ceh{6}{9}\ceb{7}{9}\cedh{8}{9} 
        \cew{1}{10}\cec{2}{10}{$\#_{1}$}\cec{3}{10}{$\#_{1}$}\cec{4}{10}{$\#_{1}$}\cec{5}{10}{$\#_{0}$}\ceh{6}{10}\cedv{7}{10} 
        \cew{1}{11}\cec{2}{11}{$\#_{1}$}\cec{3}{11}{$\#_{1}$}\cec{4}{11}{$\#_{1}$}\cec{5}{11}{$\#_{1}$}\ceh{6}{11} 
        \cew{1}{12}\cec{2}{12}{$\#_{1}$}\cec{3}{12}{$\#_{1}$}\cec{4}{12}{$\#_{1}$}\ceh{5}{12}\cedv{6}{12} 
        \cew{1}{13}\cec{2}{13}{$\#_{1}$}\cec{3}{13}{$\#_{1}$}\ceh{4}{13}\cedv{5}{13} 
        \cew{1}{14}\cec{2}{14}{$\#_{1}$}\ceh{3}{14}\cedv{4}{14} 
        \cew{1}{15}\ceh{2}{15}\cedv{3}{15} 
      \end{tikzpicture}
    \end{minipage}
  \end{center}
  \caption{Example of space-time diagram of $F_M$ (on the right) simulating the Minsky machine $M$ started on an empty counter (on the left). To simplify, we supposed $K=1$ and any state $(\#_\alpha,\delta)$ is represented as $\#_\alpha$.}
\end{figure}

The transition rule ${f_M:Q^V\rightarrow Q}$ of $F_M$ is defined as follows (Figure~\ref{fig:deffm} illustrates almost all transitions below): 
\begin{numcases}{f_M(x,y,z) =}
  h &\text{ if ${y=h}$ or (${z=h}$ and ${y\in (\{\#_1\}\times\{-1,0,1\})^k}$},\label{lt:halt}\\
  q &\text{ else if ${y=b}$ and ${x = (q,\cdot) \in Q_M\times\{-1,0,1\}^k}$},\label{lt:start}\\
  b &\text{ else if ${y=b}$ and ${x\not\in Q_M\times\{-1,0,1\}^k}$},\label{lt:blank}\\
  w &\text{ else if ${y\neq b}$ and ${x = (q,\cdot) \in Q_M\times\{-1,0,1\}^k}$},\label{lt:protect}\\
  (\#_{\min(1,\alpha+1)},\delta_i)_{1\leq i\leq k} & \text{ else if $y=(\#_\alpha,\delta_i)_i$ and $x=w$},\label{lt:leftwall}\\
  \bigl(\alpha((c_i,\delta_i),d_i)\bigr)_{1\leq i\leq k} &\text{ else if ${y=(c_i,\delta_i)_{i}\in Q_c^k}$ and ${{x=(d_i,\delta'_i)_{i}\in Q_c^k}}$},\label{lt:count}\\
  \tau\bigl(q,(v(c_i))_{1\leq i\leq k}\bigr) &\text{ else if $y=q$ and $x=(c_i,\delta_i)_{1\leq i\leq k}$},\label{lt:trans}\\
  \bigl(\beta(c_i,\delta_i)\bigr)_{1\leq i\leq k} &\text{ else if $y=(q,(\delta_i)_{1\leq i\leq k})$ and $x=(c_i,\delta'_i)_{1\leq i\leq k}$},\label{lt:unique}\\
  i_{n+1} &\text{ else if $x=y=i_n$, $0\leq n < K$},\label{lt:forcei}\\
  (q_0,0^k) &\text{ else if $y=i_K$ and $z=b$},\label{lt:istart}\\
  (\#_{-1},0)^k &\text{ else if $x=w$ and $y\in Q_M\times\{-1,0,1\}^k$,}\label{lt:icounter}\\
  w &\text{ else, \label{lt:walldefault}}
\end{numcases}

where ${\alpha : Q_c\times\{1,\#_{-1},\#_0,\#_1\}\rightarrow Q_c}$ is the counter update map defined by:
\begin{align*}
  \alpha((\#_{-1},\delta),x) &= (\#_{0},\delta),\\
  \alpha((\#_{0},\delta),x) &= (\#_{1},\delta),\\
  \alpha((\#_{1},\delta),x) &= (\#_{1},\delta),\\
  \alpha((1,\delta),x) &=
                         \begin{cases}
                           (\#_{-1},\delta) &\text{ if $x=\#_{\delta}$},\\
                           (1,\delta) &\text{ else}
                         \end{cases}
\end{align*}
and ${v:\{1,\#_{-1},\#_0,\#_1\}\rightarrow\{0,1\}}$ is a counter positivity test map defined by: 
\[v(x) =
  \begin{cases}
    1&\text{ if $x=1$},\\
    0&\text{ else}
  \end{cases}
\]
and ${\beta: Q_c\rightarrow Q_c}$ is a counter copy map defined by: 
\[\beta(x,\delta) =
  \begin{cases}
    (\#_{-1},\delta) &\text{ if $x=\#_{-1}$ and $\delta=-1$},\\
    (\#_{-1},\delta) &\text{ if $x=\#_{0}$ and $\delta<1$},\\
    (1,\delta) &\text{ else.}\\
  \end{cases}
\]

A simple verification on the transition rule gives the following.

\begin{fact}
  $F_M$ is freezing for the order $\preceq$.
\end{fact}

As said above, the encoding of the counter values of $M$ in $F_M$ is in the time interval between two particular state changes, more precisely:
for any initial configuration $c$ of $F$, the (possibly infinite) value of counter $i$ (${1\leq i\leq k}$) at position $z\in\Z$ is defined by 
\[V_i(c,z) = \bigl|\{t\in\N : F^t(c)_z = (c_j,\delta_j)_j\in Q_c^k\text{ and }c_i=1 \}\bigr|.\]

For this definition to be useful in $F_M$ we need some hypotheses on the sequence of states at the considered position: we say that the trace at ${z\in\Z}$ is valid if it contains for each counter $i$ the subsequence ${\#_{-1},\#_0,\#_1}$ and does not contain the state $w$.

The $M$-state at position $z$ starting from $c$ is ${Q(c,z)=q}$ where $(q,(\delta_i)_{1\leq i\leq k})$ is the unique occurrence at cell $z$ of a state from $Q_M\times\{-1,0,1\}^k$ (uniqueness comes from the freezing order $\preceq$ and case \ref{lt:unique} of the transition rule), and it is undefined if there is no such occurrence.

The next lemma shows that on well-formed columns, $F_M$ correctly computes a transition of $M$ from one column to the next, via the above encoding. 

\begin{lemma}[Correct simulation]
  \label{lem:correctsim}
  Let $c$ be a configuration of $F_M$ and ${z\in\Z}$ such that ${c_z\in {Q_M\times\{-1,0,1\}^k}}$ and ${c_{z+1}=b}$ and the trace at $z$ is valid. Let ${(q,(\delta_i)_{1\leq i\leq k}) = \tau\bigl(Q(c,z),(\min(1,V_i(c,z)))_{1\leq i\leq k}\bigr)}$. Then the trace at ${z+1}$ is valid, ${Q(c,z+1)}$ is well-defined and equal to $q$ and we have \[V_i(c,z+1) = \max(0,V_i(c,z) + \delta_i)\] for all ${1\leq i\leq k}$.
\end{lemma}
\begin{proof}
  ${F_M(c)_z\in Q_c}$ and ${F_M(c)_{z+1}=Q(c,z)}$ (by case \ref{lt:start} of $f_M$). Then ${F_M^2(c)_{z+1}=(q,(\delta_i)_i)}$ by case \ref{lt:trans} of $f_M$ and ${F_M^3(c)_{z+1}=(c_i,\delta_i)_{1\leq i\leq k}}$ where for each $i$:
  \begin{itemize}
  \item either $c_i=\#_{-1}$ and thus $V_i(c,z+1)=0$, but the definition of $\beta$ ensures that either $V_i(c,z)=0$ and ${\delta_i<1}$, or ${V_i(c,z)=1}$ and ${\delta_i=-1}$;
  \item or $c_i=1$ and then case \ref{lt:count} of the definition of $f_M$ and the validity of the trace at $z$ ensures that counter $i$ at position ${z+1}$ will turn into state ${(\#_{-1},\delta_i)}$ at times ${t_i+2+\delta_i}$ where $t_i$ is the time of occurrence of $\#_{-1}$ in counter $i$ at position $z$ (by definition of the map $\alpha$).
  \end{itemize}
  In any case we have a valid trace at ${z+1}$ and ${V_i(c,z+1) = \max(0,V_i(c,z) + \delta_i)}$. 
\end{proof}

\begin{lemma}[Correct halting information]
  Let $c$ be a configuration of $F_M$ and ${z\in\Z}$ a position such that ${c_z=(q_0,(0,\ldots,0))}$ and ${c_{z+1}=b}$ and the trace at $z$ is valid. Then ${F_M^t(c)_z=h}$ for some $t\in\N$ implies that the counter machine $M$ halts on input ${(V_i(c,z))_{1\leq i\leq n}}$.
  \label{lem:correcthalt}
\end{lemma}
\begin{proof}
  Suppose that $M$ doesn't halt on input ${(V_i(c,z))_{1\leq i\leq n}}$. Using Lemma~\ref{lem:correctsim}, we prove by induction on $n\in\N$ that if traces at position ${z,z+1,\ldots,z+n}$ are valid then ${Q(c,z+n)\neq h}$. Indeed it is true for ${n=0}$ and by induction we prove that if ${z+n}$ has a valid trace then ${c_{z+n}=b}$ (otherwise case \ref{lt:protect} would produce state $w$ contradicting validity). It implies that case \ref{lt:start} cannot produce state $h$ at position ${z+n}$, and therefore $h$ appears at position ${z+n}$ only if it appears at position ${z+n+1}$ (case \ref{lt:halt} is the only remaining case of $f_M$ to produce $h$). So if all traces ${z+n}$ for ${n\in\N}$ are valid then there is no $t$ such that ${F^t(c)_z=h}$. Suppose now that some trace is not valid and take ${n\geq 1}$ minimum so that the trace at ${z+n}$ is not valid. Since the trace at ${z+n-1}$ is valid and has a well-defined $Q(c,z+n-1)$ it means that there is some $t$ such that ${F_M^t(c)_{z+n-1}\in Q_s}$. Then there are two cases:
  \begin{itemize}
  \item either ${F_M^t(c)_{z+n}=b}$ but then Lemma~\ref{lem:correctsim} would apply to configuration $F_M^t(c)$ and position ${z+n-1}$ showing that the trace is valid at position ${z+n}$ which contradicts the hypothesis on $n$;
  \item or ${F_M^t(c)_{z+n}\neq b}$ so by case \ref{lt:protect} of $f_M$ we would have ${F_M^{t+i}(c)_{z+n}=w}$ for all ${i\geq 1}$ showing that $h$ does not appear at position ${z+n}$.
  \end{itemize}
  The lemma follows.
\end{proof}

\begin{theorem}
  There exists a 1D freezing CA $F$ such that $\CYREACH{F}$ is undecidable.
  \label{thm:freezecyreach}
\end{theorem}
\begin{proof}
  Let $M_u$ be the machine from Theorem~\ref{theo:minsky}. We claim that ${\CYREACH{F_{M_u}}}$ is undecidable because the halting problem of $M_u$ on a given input reduces to it. Indeed, given an input ${(\chi_i)_{1\leq i\leq k}}$, let ${l=\max_i\chi}$ and define ${v\in Q^{\ball{l+2}}}$ by: 
    \[v_j =
      \begin{cases}
        w &\text{ if $j=-l-2$},\\
        (\alpha_{-j,i})_{1\leq i\leq k} & \text{ if $-l-1\leq j<0$},\\
        (q_0,(-1,\ldots,-1)) & \text{ if $j=0$},\\
        b &\text{ if $0<j\leq l+2$,}
      \end{cases}
    \]
    where ${\alpha_{j,i}\in Q_c}$ is ${(1,-1)}$ if ${j-1<\chi_i}$ and ${(\#_{-1},0)}$ else. One can check that for any ${c\in[v]}$ and any position $z$ with ${-l-1\leq z\leq l+2}$ the trace at $z$ in $c$ is valid: it follows from the definition of $f_M$ for ${-l-1\leq z\leq 0}$ (case \ref{lt:leftwall} at ${z=-l-1}$ and \ref{lt:count} elsewhere) and for ${z>0}$ it follows from Lemma~\ref{lem:correctsim}. Moreover, by choice of $v$ and for any $i$ (${1\leq i\leq k}$), we have ${V_i(c,-\chi_i-1)=0}$ and ${V_i(c,-\chi_i-1+p)=p}$ for all ${1\leq p\leq \chi_i}$ by a straightforward induction. Therefore we have ${V_i(c,-1)=\chi_i}$ and ${V_i(c,0)=\chi_i}$ by choice of $v$. We deduce from Lemma~\ref{lem:correcthalt} that cylinder ${[h]}$ can be reached from cylinder ${[v]}$ only if $M_u$ halts on input ${(\chi_i)_i}$. Finally, by choosing ${c\in[v]}$ such that ${c_z=b}$ for ${z>l+2}$, it follows from Lemma~\ref{lem:correctsim} that if $M_u$ halts on ${(\chi_i)_i}$ then ${F_M^t(c)_0=h}$ for some $t$. This proves the reduction of the halting problem of $M_u$ to ${\CYREACH{F_{M_u}}}$ and the theorem follows.
\end{proof}

The same construction allows to prove the following result which might seem surprising at first: even if it is easy to check that a CA is freezing (Fact~\ref{fact:freezingdecidable}), and if an immediate bound on the number of changes in a freezing CA is given by the size of the alphabet, it is still undecidable to determine what is the actual maximal number of changes.

\begin{theorem}
  There exists $k\in\N$ such that it is undecidable to determine whether a given freezing CA $F$ is $k$-change.
\end{theorem}
\begin{proof}
  Let $M$ be any counter machine and $F_M$ the associated freezing CA. It follows from the freezing order $\preceq$ that the unique sequence of state changes giving ${K+5}$ changes is (repetition of a same state removed): 
  \[i_0\rightarrow\cdots i_K\rightarrow (q_0,0^k)\rightarrow (\#_{-1},0)^k\rightarrow (\#_{0},0)^k\rightarrow (\#_{1},0)^k\rightarrow h,\]
and any other sequence compatible with $\preceq$ has strictly less state changes. Indeed, it is the longest among the sequences containing $i_K$ and the choice of $K=3k+3$ ensures that all $\preceq$-admissible sequences without occurrence of $i_K$ have strictly fewer than $K$ changes.

Consider a configuration $c$ of $F_M$ such that the maximal sequence of state changes above occurs at cell $0$. By case \ref{lt:forcei}, $c$ must be such that ${c_z=i_0}$ for any ${-K+1\leq z\leq 0}$ otherwise $w$ would appear (by case \ref{lt:walldefault}) at position $0$. Moreover, we must have $c_1=b$ and otherwise we would have ${F^{K+1}(c)_0=w}$. Therefore we have ${F^{K+1}(c)_0=(q_0,0^k)}$ and ${F^{K+1}(c)_{-1}=w}$ (by cases \ref{lt:istart} and \ref{lt:walldefault}). Moreover cases \ref{lt:icounter} and \ref{lt:blank} force ${F^{K+1}(c)_1=b}$. Finally, Lemma~\ref{lem:correcthalt} can be applied at position $0$ of configuration ${F^{K+1}(c)}$ showing that $h$ appears at position $0$ only if $M$ halts on empty input (because the state change sequence at position $0$ is such that $V_i(c,0)=0$ for any ${i}$). Conversely, if $M$ halts on the empty input and if we consider configuration $c$ such that ${c_z=i_0}$ for ${z\leq 0}$ and ${c_z=b}$ for ${z>0}$, it is easy to check (using Lemma~\ref{lem:correctsim}) that the sequence of state changes at position $0$ in $c$ is exactly the above sequence. We conclude that $M$ halts on the empty input if and only if ${F_M}$ is not ${K+4}$-change. The theorem follows.
\end{proof}

\section{Limit Fixed-Point Computability}
\label{sec:limits}

In the previous section we studied canonical problems of prediction and reachability defined for any cellular automaton. Here we consider problems related to the convergence property. We can first consider the limit value of a cell, and the time its takes to reach this limit given an initial configuration.

\begin{definition}
  Let $F$ be any convergent CA. The limit value problem ${\LIMIT{F}}$ is defined as follows:
  \begin{itemize}
  \item input: a computable configuration $c$ (given as a Turing machine);
  \item output: ${F^\omega(c)_0}$.
  \end{itemize}
The freezing time problem ${\TIME{F}}$ is defined as follows:
  \begin{itemize}
  \item input: a computable configuration $c$ (given as a Turing machine);
  \item output: the freezing time $\fzt{c,0}$ of cell $0$ starting from $c$.
  \end{itemize}
\end{definition}

\begin{remark}
  Problems ${\LIMIT{F}}$ and ${\TIME{F}}$ are Turing-equivalent for any freezing CA $F$ since we have ${F^\omega(c)_0 = F^{\fzt{c,0}}(c)_0}$ and ${\fzt{c,0}=\min\{t:F^t(c)_0=F^\omega(c)_0\}}$.
\end{remark}

\begin{theorem}
  \label{thm:uniflimitfreezing}
  There is a freezing CA $F$ such that both $\TIME{F}$ and $\LIMIT{F}$ are uncomputable.
\end{theorem}
\begin{proof}
  From the remark above it is sufficient to prove uncomputability of $\LIMIT{F}$. The example $F_M$ of Theorem~\ref{thm:freezecyreach} has the desired property and the proof of the theorem actually shows it because:
  \begin{itemize}
  \item it shows undecidabilty of reachability of cylinder $[h]$ where $h$ is an invariable state of $F_M$, therefore it shows that given some bounded configuration $u$ it is undecidable whether there is ${c\in[u]}$ such that ${F^\omega(c)_0=h}$;
  \item there is in fact a canonical computable (actually ultimately constant) configuration $c_u$ such that $[h]$ is reachable from $[u]$ if and only if ${F^\omega(c_u)_0=h}$.
  \end{itemize}
\end{proof}

In the following subsection we are going to consider a non-uniform problem on the limit fixed-point for which a difference between bounded-change and convergent CA will appear. Let us first give some computability upper bounds for such limit fixed-points.

Given a configuration $c\in\Z^d$ and a state $q$, denote by ${\chi_q(c)}$ the set of positions in $\Z^d$ which are in state $q$ in $c$:
\[\chi_q(c) = \bigl\{z\in\Z^d : c_z=q\bigr\}.\]
The following proposition gives immediate upper-bounds on such characteristic sets for limit fixed points in terms of the arithmetical hierarchy \cite{Rogers}.

\begin{proposition}
  \label{prop:maxlimitcomplexity}
  Let $F$ be a convergent CA and ${c\in\Z^d}$ be any computable configuration. For any state $q$ the set ${\chi_q(F^\omega(c))}$ is a $\Delta^0_2$ arithmetical set. Moreover, if $F$ is freezing with order $\preceq$ then:
  \begin{itemize}
  \item ${\chi_q(F^\omega(c))}$ is recursively enumerable if $q$ is $\preceq$-minimal;
  \item ${\chi_q(F^\omega(c))}$ is co-recursively enumerable if $q$ is $\preceq$-maximal.
  \end{itemize}
\end{proposition}
\begin{proof}
  First for a convergent CA $F$ we have 
  \[\chi_q(F^\omega(c)) = \{z : \exists t_0, \forall t, t\geq t_0\Rightarrow F^t(c)_z=q\}\]
  and also 
  \[\chi_q(F^\omega(c)) = \{z : \forall t, \exists t_0, F^{t+t_0}(c)_z=q\}\]
  which shows that ${\chi_q(F^\omega(c))\in\Delta_2^ 0}$. When $F$ is freezing with order $\preceq$ and $q$ is $\preceq$-minimal we have 
  \[\chi_q(F^\omega(c)) = \{z : \exists t, F^t(c)_z=q\}\] and if $q$ is $\preceq$-maximal then 
  \[\chi_q(F^\omega(c)) = \{z : \forall t, F^t(c)_z=q\}.\]
\end{proof}

\subsection{1D bounded-change CA}

In dimension 1, bounded-change CA are too much restricted to produce uncomputable limits from computable initial configurations. The argument in the following theorem is due to G. Richard.

\begin{theorem}
  \label{thm:computablelimits}
  For any 1D bounded-change $F$ and any computable configuration $c$, ${F^\omega(c)}$ is computable.
\end{theorem}
\begin{proof}
  First we can suppose without loss of generality that $F$ has radius $1$ (if $F$ has radius $r$ and alphabet $Q$ one can consider $F'$ of radius $1$ and alphabet $Q^r$ obtained from $F$ by grouping cells by blocks of size $r$, computability of configuration is not affected by this grouping operation). Now let $c$ be a fixed computable configuration and denote by ${\lambda(z)}$ the number of changes occurring at position $z$ in the orbit of $c$: 
  \[\lambda(z) = \bigl|\{t : F^{t+1}(c)_z\neq F^t(c)_z\}\bigr|.\]

  We claim that there is an algorithm that, given ${z\leq z'}$ and ${\lambda(z)}$ and ${\lambda(z')}$, correctly computes ${F^\omega(c)_{[z,z']}}$. Indeed, knowing ${\lambda(z)}$ and ${\lambda(z')}$, we can compute ${F^t(c)_{[z,z']}}$ with $t$ the first time such that ${F^t(c)_z = F^\omega(c)_z}$ and ${F^t(c)_{z'} = F^\omega(c)_{z'}}$ (it is sufficient to compute more and more time steps until having observed ${\lambda(z)}$ changes at $z$ and ${\lambda(z')}$ changes at $z'$). Then, since the states of cells $z$ and $z'$ do not change after time $t$, and since $F$ has radius $1$, ${F^\omega(c)_{[z,z']}}$ only depends on ${F^t(c)_{[z,z']}}$ and more precisely ${F^\omega(c)_{[z,z']} = F^{t+k(z'-z)}(c)_{[z,z']}}$ if $F$ is $k$-change.

  Let us define the following constants, depending\footnote{They can depend on $c$ in a non-recursive way. In fact the remainder of the proof shows that they cannot depend recursively on $c$ otherwise it would contradict Theorem~\ref{thm:uniflimitfreezing}.} on $c$, which are the maximal number of changes occurring infinitely often to the left/right respectively, and the positions beyond which these maxima are bounds on the number of changes:
  \begin{align*}
    L &= \max \{i : \forall z<0,\exists z'<z, \lambda(z')=i\}\\
    z_L &= \max \{z : \forall z'<z, \lambda(z')\leq L\}\\
    R &= \max \{i : \forall z>0,\exists z'>z, \lambda(z')=i\}\\
    z_R &= \min \{z : \forall z'>z, \lambda(z')\leq R\}\\
  \end{align*}
  
  The algorithm to compute ${F^\omega(c)_z}$ given $z$ is the following: compute larger and larger portions of the space-time diagram around position $z$ until finding ${z_1\leq z \leq z_2}$ and $t$ such that:
  \begin{enumerate}
  \item ${z_1\leq z_L}$ and ${z_2\geq z_R}$,
  \item the state of $z_1$ has changed $L$ times before time $t$ and the state of $z_2$ has changed $R$ times before time $t$.
  \end{enumerate}
  By definition of ${L,z_L,R,z_R}$, such values ${z_1,z_2}$ and $t$ can always be found for any $z$ and we have ${\lambda(z_1)=L}$ and ${\lambda(z_2)=R}$. Then it is sufficient to apply the algorithm of the above claim to compute ${F^\omega(c)_{[z_1,z_2]}}$ and therefore obtain ${F^\omega(c)_z}$.
\end{proof}

\subsection{1D convergent CA}

Contrary to bounded-change CA, convergent CA can have arbitrarily many changes at a given cell depending on the context. This is enough to obtain uncomputable limit fixed point from computable initial configurations as shown by the following CA $F$. Intuitively, $F$ is a 1D CA that simulates progressively all Turing machines for more and more time steps and, each time a machine $i$ halts, the head travels to a predetermined computable position $p(i)$ to write a mark, and then goes back to the simulation zone and goes on simulating other Turing machines. Any position $p(i)$ would ultimately contain the information of whether machine $i$ halts or not. The Turing simulation is done in a finite zone that grows as needed but also that moves regularly to the right, therefore any given cell is out of this active computation zone after a finite time. The position map $p(i)$ is increasing so that less than $i$ positions lie to the left of $p(i)$, which means that $p(i)$ will be crossed at most $i$ times by the marking mechanism. This intuitive description is probably enough to be convinced that there is a CA $F$ and a computable configuration $c$ such that ${\bigl(F^t(c)\bigr)_t}$ converges to an uncomputable configuration. However, we want $F$ to be a convergent CA, \textit{i.e.} we want the convergence of ${\bigl(F^t(c)\bigr)_t}$ for any inital configuration $c$, which requires a much more careful design of $F$. In what follows it is important to keep in mind that the construction is not sensitive to details in the implementation of the Turing machines: in fact the property of $F$ being convergent is independent of the actual Turing computation. That's why we focus on the marking process.

The general idea is to implement an addressing mechanism working in unary: positions to be marked are special zones of a certain length, and the marking process consists in a 'snake' of some length moving to the left until it reaches a special zone of matching length. The convergence is guaranteed by a counter on each special zone that allow only a bounded number of successive 'snakes' to cross the zone. A key aspect of the construction to help the analysis and the proof of convergence is the presence of a unique global head that can either work on the Turing computation or be the 'head of the snake' in the marking process. Essentially all state changes occur in the neighborhood of the head, so convergence is guaranteed by the fact that this global head never visits infinitely often the same cell.

\newcommand\prep{\texttt{Preparing}}
\newcommand\mrk{\texttt{Marking}}
\newcommand\retur{\texttt{Return}}
\paragraph{Global structure.} $F$ has alphabet ${Q=Q_T\times Q_M\times Q_A\times Q_H\cup\{e\}}$ where $e$ is a special error state that spreads over the entire configuration as soon as it appears, and all normal states have four components:
  \begin{itemize}
  \item the \emph{Turing component} $Q_T = \{L,R\}\cup E_T$ which holds the simulation of Turing machines;
  \item the \emph{marker component} $Q_M=\{0,1\}$ which is used to mark predetermined positions $p(i)$ for halting machines $i$;
  \item the \emph{addressing component} $Q_A=\{0,1\}\cup\{0,1\}^4$ which is used to handle the low level mechanism for the marking head to stop at the correct position;
  \item the \emph{global head component} $Q_H = \{L,R\}\cup S_H$ which ensures unicity of a global head and controls the succession of phases of Turing simulation and marking process.
  \end{itemize}

The global head is not the Turing head which is encoded by states in $E_T$.

\paragraph{Addressing component.} This component is a kind of ``conveyor belt'' with zigzags that globally shift information as follows:
\begin{itemize}
\item A \emph{zigzag} is any maximal connected zone in state
  ${\{0,1\}^4}$. More precisely, states from ${\{0,1\}}$ correspond to
  a simple belt and states from ${\{0,1\}^4}$ correspond to three
  layers of belt forming a zigzag (plus a fourth layer used as a
  passage counter for the zigzag). Each position in this virtual belt
  has a well-defined predecessor and successor as depicted in
  Figure~\ref{fig:zigzagbelt}. Note that in some degenerate cases,
  this virtual belt has several connected components but always finite ($3$ in the case where all cells with possibly finite exceptions are in a state from $\{0,1\}^4$ in their addressing component, $1$ in any other case). 
\item The default dynamics of $F$ on this addressing component is simply that each bit $b$ at any position in the virtual belt is shifted to its successor's position. This default dynamics can be changed only when the global head is near (details below).
\end{itemize}
\newcommand\case[2]{\draw (#1,#2)++(-.5,-.5) rectangle +(1,1);}
\newcommand\caseg[2]{\draw[fill=gray] (#1,#2)++(-.5,-.5) rectangle +(1,1);}

\begin{figure}
  \centering
    \begin{minipage}{.4\linewidth}
      \begin{center}\tiny
        ${\{0,1\}\ \{0,1\}\ \{0,1\}^4\ \{0,1\}^4\ \{0,1\}^4\ \{0,1\}\ \{0,1\}^4\ \{0,1\}\ \{0,1\}}$
      \end{center}
    \end{minipage}
    \begin{minipage}{.4\linewidth}
      \begin{center}
        \begin{tikzpicture}[scale=.5,>=stealth]
          \case{-1}{0}\case{7}{0}
          \case{0}{0}\case{1}{0}\case{2}{0}\case{3}{0}\case{4}{0}\case{5}{0}\case{6}{0}
          \case{1}{1}\case{2}{1}\case{3}{1}
          \case{1}{-1}\case{2}{-1}\case{3}{-1}
          \caseg{1}{-2}\caseg{2}{-2}\caseg{3}{-2}
          \case{5}{1}\case{5}{-1}\caseg{5}{-2}
          \draw[->] (6,0)--(5,1);
          \draw[->] (5,-1)--(4,0);
          \draw[->] (5,1)--(5,0);
          \draw[->] (5,0)--(5,-1);
          \draw[->] (4,0)--(3,1);
          \draw[->] (3,1)--(2,1);
          \draw[->] (2,1)--(1,1);
          \draw[->] (1,1)--(1,0);
          \draw[->] (1,0)--(2,0);
          \draw[->] (2,0)--(3,0);
          \draw[->] (3,0)--(3,-1);
          \draw[->] (3,-1)--(2,-1);
          \draw[->] (2,-1)--(1,-1);
          \draw[->] (1,-1)--(0,0);
          \draw[->] (0,0)--(-1,0);          
          \draw[->] (7,0)--(6,0);          
        \end{tikzpicture}
      \end{center}
    \end{minipage}
  \caption{Conveyor belt and zigzags in the addressing component.}
  \label{fig:zigzagbelt}
\end{figure}
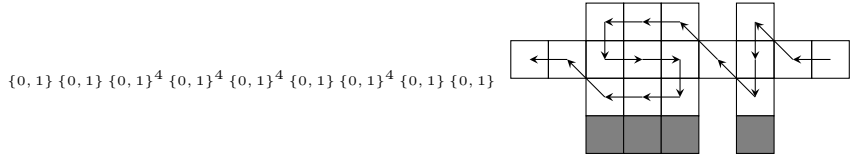

\paragraph{Validity checks and computation cycle.} We call \emph{valid} a configuration $c$ whose orbit contains no occurrence of $e$. By the spreading behavior of $e$ we have that any orbit which is not valid is convergent. $F$ constantly checks and maintains the following conditions and produces the error state $e$ in case of violation:
\begin{itemize}
\item in Turing component the following patterns are forbidden: $LR$, $RL$, $RE_T$, $E_TL$, so that a valid configuration contains at most one connected zone in alphabet $E_T$, called \emph{Turing zone};
\item in head component the following patterns are forbidden: $LR$, $RL$, $RS_H$, $S_HL$, $S_HS_H$, so that a valid configuration contains at most one state in $S_H$;
\item in the addressing component the following patterns are forbidden: ${(a,b,c,1)(d,e,f,0)}$, so that in a valid configuration the fourth layer of any zigzag holds the unary representation of a counter (distance of the first 0 to the rightmost position of the zigzag) between 0 and the length of the zigzag, called \emph{zigzag counter};
\item in the addressing component, when a position in the virtual belt is in state $0$ then its predecessor must also be in state $0$, except if the head is present on the predecessor's position (precise definition below); therefore each connected component of the virtual belt seen as a bi-infinite binary word contains a single connected component of $1$s;
\item in the addressing component, no zigzag can accept a too long connected component of $1$ entering the zigzag, concretely: if there is a 1 in the third layer of the rightmost position of a zigzag, then there must be a 0 on the first layer at the same position, otherwise error state $e$ is generated; this prevents any connected component of $1$s of length ${\geq 2n}$ initially to the right of a zigzag of length $n$ to cross it.
\end{itemize}

The state of the unique global head (if present) is one of the following: \[S_H=\{\leftarrow,\rightarrow,\prep_1,\prep_3,\mrk_1,\mrk_2,\mrk_3,\retur\}.\] Each one corresponds to a a specific phase of the behavior of $F$ with additional validity checks that generate $e$ in case of failure. Intuitively, $\leftarrow$ and $\rightarrow$ correspond to back and forth movements inside the Turing zone producing one step of Turing computation at each back and forth. Then states $\prep_x$, $\mrk_x$ and $\retur$ correspond to an interruption of the Turing computation to do a marking cycle: preparation of marking address, movement and marking at the correct position, return to Turing computation. The precise position of the head on the virtual belt is as follows: if the head is in state $\prep_x$ or $\mrk_x$ then it is on layer $x$ of the belt in the cell, otherwise it is on layer $1$. In a valid configuration the sequence of states taken by the global head is always a subword of the language 
\[\bigl(\rightarrow^+\leftarrow^+((\prep_1^++\prep_3^+)(\mrk_1+\mrk_2+\mrk_3)^+\retur^+)?\bigr)^+\]
\begin{itemize}
\item $\leftarrow$: the global head moves back to the leftmost position of the Turing zone; the head must be inside the Turing zone; when this leftmost position is reached, the global head turns into state $\rightarrow$;  moreover, if during the crossing of the Turing zone a special interruption state of the Turing head is encountered (in $E_T$), then the global head turns into state $\prep_3$ if the position is inside a zigzag and $\prep_1$ else;
\item $\rightarrow$: the global head runs through the entire Turing zone from left to right while shifting it one cell to the right, and applying one step of the Turing computation, and adding one more cell to the zone at the right end; the global head must be inside the Turing zone; when the right end is reached and the new cell added to the Turing zone, the global head turns into state $\leftarrow$;
\item $\prep_x$: the global head starts to write $1$s on the addressing component while moving two cells left at each step and until it meets the left boundary of the Turing zone; more precisely, when in state $\prep_1$ it writes $1$s onto the first layer if inside a zigzag and onto the third when in state $\prep_3$; outside zigzags it always write on the unique layer present\footnote{Note that due to the checks done on the addressing component (see above), this $\prep_x$ stage can generate an $e$ state: for instance if the head in state $\prep_3$ exits a zigzag and later enters into a new one, the $1$s written on the virtual belt will no longer be connected. This kind of degenerate behavior can easily be avoided by a sufficient spacing between zigzags.}; the global head must be inside the Turing zone and when it reaches the left boundary of the Turing zone, it turns into state $\mrk_1$; the purpose of this stage is to write a connected component of $1$s of length $n/2$ onto the virtual belt where $n$ is the distance between the global head and the leftmost position of the Turing zone when the global state changes from $\leftarrow$ to $\prep_x$;
\item $\mrk_x$: the global heads moves along the virtual belt from a position to its successor; when inside a zigzag the $x$ indicates which of the $3$ layers of the belt the global head is currently on; the position on the belt preceding that of the global head must contain a $1$ and the successor must contain a $0$, this forces a connected component of 1s that moves along the belt behind the head which can be seen as a unary value called \emph{the address value}; when inside a zigzag the global head does the following:
  \begin{itemize}
  \item when entering a zigzag it checks that the zigzag counter is not $0$ otherwise the error state $e$ is generated;
  \item when reaching the rightmost position of the second layer (state $\mrk_2$), it checks whether the address value is exactly twice the width of the zigzag, precisely: if the first layer contains a 1 at its position and the cell immediately to its right contains value $0$, then it enters into state $\retur$ and a 1 is written on the marker component;
  \item when in state $\mrk_1$ on the first layer of a zigzag and at the position of the leftmost 1 of the zigzag counter, the zigzag counter is decremented by 1 by changing the state 1 at the current position into a 0;
  \end{itemize}
\item $\retur$: the global head moves along the virtual belt backward (from a position to its predecessor) and transforms all 1s into 0s until it reaches the left boundary of the Turing zone at which point it turns into state $\rightarrow$.
\end{itemize}

Given a valid configuration $c$, we say that a component is convergent if ${\bigl(\pi\circ F^t(c)\bigr)_{t\in\N}}$ is convergent where $\pi$ is the projection of states onto the considered component. $c$ is convergent when all components are convergent.

\begin{lemma}
  \label{lem:nochangeconv}
  If $c$ is a valid configuration such that the global head state changes only a finite number of times in its orbit, then $c$ is convergent.
\end{lemma}
\begin{proof}
  We can suppose without loss of generality that the global state ${s\in S_H}$ never change (otherwise consider configuration ${F^t(c)}$ where $t$ is the time when the last state change occurs). Given the hypothesis we have the following possible cases: 
  \begin{itemize}
  \item ${s=\rightarrow}$ and the Turing zone extends infinitely to the right: in this case the head moves to the right forever so the Turing component  and the global head component are clearly convergent. The marker component is also convergent since it can only change when the global head state changes from $\mrk_x$ to $\retur$. Finally, the addressing component contains a single connected component of 1s per connected component of the virtual belt (because $c$ is valid). Each cell either stays forever inside such a connected component of 1s, or stays forever outside after a finite time. In any case, the addressing component is convergent.
  \item ${s=\leftarrow}$ and the Turing zone extends infinitely to the left: this case is symmetric to the previous one.
  \item ${s=\prep_x}$ and the Turing zone extends infinitely to the left: the Turing component and the marker component don't change so they are convergent. Moreover the global head moves to the left forever, so the global head component is also convergent. The global head in this case constantly extends to the left a connected component of 1s while moving to the left. Therefore, for any cell, after a finite time the global head is on the left and will never come back leaving a similar situation as above: the addressing component is again convergent.
  \item ${s=\mrk_x}$: the global head moves to the left (${x=1}$ or $3$) or to the right (${x=2}$) forever followed by a (possibly infinite) snake of 1s. The situation is similar to the above cases and the convergence follows for the same reasons: the global heads moves towards infinity, nothing is changed on the Turing, head and marker components, and the addressing component converges because of the uniqueness of the component of 1s (per connected component of the virtual belt).
  \item ${s=\retur}$: the global head moves backward on the virtual belt, so it visits each cell a finite number of times (at most $3$ in the case of a cell inside a finite zigzag). After this finite number of visits the state of a cell converges because nothing is changed on the Turing, head and marker components, and the addressing component converges as in the previous case.
  \end{itemize}
\end{proof}

\begin{lemma}
  \label{lem:infchangeconv}
  If $c$ is a configuration such that the global head state changes infinitely many times, then $c$ is convergent.
\end{lemma}
\begin{proof}
  Consider any position ${z\in\Z}$, we will show that ${\bigl(F^t(c)_z\bigr)_{t\in\N}}$ is convergent. First remark that such a $c$ must be valid. Moreover the Turing component is convergent because:
  \begin{itemize}
  \item either states $\leftarrow$ and $\rightarrow$ do not occur after some finite time in which case the Turing component becomes constant;
  \item or there are infinitely many state changes from ${\leftarrow}$ to ${\rightarrow}$ or from ${\retur}$ to ${\rightarrow}$ which implies infinitely many shifts to the right of the computation zone;
  \item other cases are impossible by the hypothesis on $c$.
  \end{itemize}
  Considering the sequence of global head states, we must be in one of the following situations: 
  \begin{itemize}
  \item finitely many ${s\in\{\mrk_x,\retur,\prep_x\}}$: in this case after some finite time $T$, the global head is restricted to the Turing zone which moves constantly to the right, therefore the marker and global head components are convergent. Moreover, the addressing component evolves independently of the global head after time $T$ and, since $c$ is valid, there is only one connected component of 1s per connected component of the virtual belt. The convergence follows.
  \item the state of the global head eventually stays in ${\mrk_1,\mrk_2,\mrk_3}$: in this case, after some finite time the global head moves along the virtual belt forever. It cannot be inside or enter an infinite zigzag (otherwise it would eventually stay in a fixed state forever), so it escapes to the left, precisely after some time it stays forever to the left of position $z$. This shows that the Turing, marking and global head components are convergent at $z$. By validity of $c$ and the same reasoning as above about connected components of 1s, the addressing component converges at $z$.
  \item infinitely many changes from $\mrk_x$ to $\retur$: the global head can only do such state changes inside finite zigzags, because on one hand this state change can only occur at the right boundary of a zigzag, and, on the other hand, if the head turns into state $\retur$ inside a zigzag which extends infinitely to the left then it could never come back to the leftmost position of the Turing zone, contradicting the infinite state changes hypothesis. The global head can visit only finitely many times each such finite zigzags, therefore after some time it will stay forever to the right of position $z$. With the same reasoning as in the previous case, all components are convergent at position $z$.
  \end{itemize}
\end{proof}

\newcommand\mach{\textsf{Machines}}

\begin{theorem}
  \label{thm:uncomputablelimit}
  There is a convergent CA $F$ and a computable configuration $c$ such that ${F^\omega(c)}$ is uncomputable.
\end{theorem}
\begin{proof}
  Lemma~\ref{lem:nochangeconv} and~\ref{lem:infchangeconv} show that the construction technique above always produces convergent CA whatever the behavior of the Turing component. Let's define $F$ as the above construction applied to a Turing machine that does the following when started on the empty tape:
  \begin{center}
    \small\tt
    \begin{itemize}
    \item ${\mach\leftarrow\emptyset}$
    \item ${n\leftarrow0}$
    \item loop forever:
      \begin{itemize}
      \item ${n\leftarrow n+1}$
      \item wait $2n$ steps
      \item ${\mach\leftarrow\mach\cup\{n\}}$
      \item for each ${i\in\mach}$ do:
        \begin{itemize}
        \item simulate $n$ steps of machine $i$
        \item if $i$ has halted during the simulation then
          \begin{enumerate}
          \item ${\mach\leftarrow\mach\setminus\{i\}}$
          \item place the Turing head at distance $4i$ from the left boundary of the computation zone
          \item turn into special interruption state to launch a marking process
          \end{enumerate}
        \end{itemize}
      \end{itemize}
    \end{itemize}
  \end{center}
  This Turing machine has the following properties:
  \begin{itemize}
  \item it launches a marking process with address value\footnote{The Turing head placement is $4i$ cells away from the left boundary of the computation zone when the interruption launches the marking process; then the preparing stage produce a sequence of 1s of length $2i$.} $2i$ if and only if machine $i$ halts on empty input, and it launches it at most once;
  \item the simulation of machine $i$ starts after strictly more than ${i(i-1)}$ steps, therefore the marking process with address value $i$ cannot be launched before time step ${i(i-1)+1}$.
  \end{itemize}
  Now consider the following computable configuration $c$:
  \begin{itemize}
  \item the Turing component corresponds to an empty tape with machine head in initial state at position $0$;
  \item the marker component is uniformly $0$;
  \item the global head component is in state $\rightarrow$ with head at position $0$;
  \item the addressing component is the concatenation for all $i$ of a zigzag of length $i$ with counter at value $i$ followed by $i$ cells of simple belt, concretely:
    \[{}^\omega 0\ (0,0,0,1)\ 0\ (0,0,0,1)^2\ 0^2\ \cdots\ (0,0,0,1)^i\ 0^i\ \cdots\]
    where the first zigzag (of length 1) is at position $0$.
  \end{itemize}
  We claim that the marker component of ${F^\omega(c)}$ contains a $1$ at position ${p(i)=i+\sum_{0<j<i}2j}$ (the rightmost position of zigzag of length $i$) if and only if machine $i$ halts on empty input. This shows that ${F^\omega(c)}$ is uncomputable and concludes the theorem. To prove the claim, it is sufficient to check that each marking process launched by the Turing computation described above works properly and maintains the validity of the configuration. 
  Starting from $c$, the definition of $F$ ensures that the marking process with address value $i$ will be launched (if ever launched) after more than ${i(i-1)}$ Turing steps and therefore in a configuration where the left boundary of the Turing zone will be at a position ${p>i(i-1)}$ (because the Turing zone moves one cell to the right at each Turing step). This guarantees that the $\prep_x$ stage of the marking process will work properly since at such a position zigzags are of length greater than $i$, and the global head will therefore enter or leave a zigzag at most once before turning into state $\mrk_x$. Then, during the $\mrk_x$ stage, the zigzag counter will never be $0$ when the global head enters because in $c$ the zigzag of length $k$ has only ${k-1}$ zigzags to its left and will therefore be crossed (or entered in) at most ${k}$ times in total. Finally the $\retur$ stages poses no problem.
\end{proof}

\subsection{2D+ freezing CA}
\label{sec:freezatam}

Let us first remark that in dimension 2 and more, it is very easy to produce uncomputable limit fixed point from computable initial configurations.

\begin{proposition}
  Let $F$ be the freezing CA defined over alphabet $Q=\{0,1\}$ by: 
  \[F(c)_z = \min(c_z,c_{z+(0,1)}).\]
  There is a computable configuration $c$ such that ${F^\omega(c)}$ is uncomputable.
\end{proposition}
\begin{proof}
  Consider the configuration $c$ defined by: 
  \[c(i,j) =
  \begin{cases}
    0 &\text{ if $j>0$ and machine $i$ halts on the empty input in less than $j$ steps}\\
    1 &\text{ else.}
  \end{cases}
  \]
  It is straighforward to check that ${F^\omega(c)_{(i,0)} = 0}$ if and only if machine $i$ halts on the empty input.
\end{proof}

We can actually obtain much stronger and meaningful results on the uncomputability of limit fixed-points for freezing CA by requiring a finite initial configuration.
Let us recall that an aTAM system is \emph{directed} if there is a unique terminal assembly. Taking the notations of example~\ref{ex:atam}, it means in particular that this unique terminal assembly is ${F_R^\omega(c_0)}$. Then the following result on the computational power of directed aTAM systems is directly related to our concern (we use again notations of example~\ref{ex:atam} to state the theorem).

\begin{theorem}[Main construction of \cite{LathropLPS11}]
  \label{thm:hardatam}
  There exists a computable function ${f:\N\rightarrow\N}$ such that for any recursively enumerable set $A$ there exists a directed aTAM whose unique terminal assembly $c$ is such that for any ${n\in\N}$: 
  \[c(f(n),0) =
    \begin{cases}
      t&\text{ if }n\in A\\
      \epsilon&\text{ if }n\not\in A\\
    \end{cases}
  \]
  where $t$ is some fixed tile of the system.
\end{theorem}

From Proposition~\ref{prop:maxlimitcomplexity}, the above result is optimal since, in directed aTAM systems and the corresponding freezing CA, each state is either minimal or maximal with respect to the order, so sets $\chi_q(F_R^\omega(c))$ are at most at level one of the arithmetical hierarchy. However for freezing CA in general we can do more as shown below.

\begin{corollary}
  \label{cor:hardlimit2Dfreezing}
  There exists a 2D freezing CA $F$, a finite configuration $c$, and a state $q$ such that  ${\chi_q(F^\omega(c))}$ is neither recursively enumerable nor co-recursively enumerable.  
\end{corollary}
\begin{proof}
  Fix some computable bijection ${\phi:\N\rightarrow\N^2}$ and consider the two following recursively enumerable sets:
    \begin{align*}
      A_1 &= \{n : \pi_1(\phi(n))\in H\},\\
      A_2 &= \{n : \pi_2(\phi(n))\in H\}
    \end{align*}
    where $\pi_1$ and $\pi_2$ are the projections on first and second components and $H$ is the halting set. From Theorem\ref{thm:hardatam} we get two directed aTAM to which correspond two freezing CA $F_1$ and $F_2$ whose limit fixed-point starting from the respective seed configurations $c_0^1$ and $c_0^2$ verify: 
    \[
      F_i^\omega(c_0^i)_{(f(n),0)} =
      \begin{cases}
        t_i&\text{ if }n\in A_i\\
        \epsilon_i&\text{ if }n\not\in A_i\\
      \end{cases}
    \]
    for ${i=1}$ or $2$. Consider now the freezing CA ${F=F_1\times F_2}$ (with the product order), the configuration $c$ whose first component is $c_0^1$ and second component is $c_0^2$, and the state ${q=(t_1,\epsilon_2)}$. One can check that ${(f(n),0)\in \chi_q(F^\omega(c))}$ if and only if ${\phi(n)=(i,j)}$ where ${i\in H}$ and ${j\not\in H}$. The corollary follows.
\end{proof}

\section{Recap of results}
\label{sec:recap}

\newcommand\lbd[1]{{\color{red}\textbf{#1}}}
\newcommand\ubd[1]{{\color{blue}\textit{#1}}}

We give in the following table a synthesis of results showing differences between freezing, bounded-change and convergent CA together with dimension sensitiveness. The table mixes different kinds of results, among which: \lbd{lower bounds L} which must be read as ``there exists $F$ such that property L holds'', and \ubd{upper bounds U} to be read as ``for all $F$ property U holds''.

\newcommand\resref[2]{
  \begin{tabular}[t]{c}
    #1\\
    {\small\it (#2)}
  \end{tabular}
}

\newcommand\upperbd[2]{\resref{\ubd{#1}}{#2}}
\newcommand\lowerbd[2]{\resref{\lbd{#1}}{#2}}

\begin{center}
  \begin{tabular}{r|c|c|c}
    &Freezing&Bounded change&Convergent\\
    \hline
    Membership & \resref{$\PTIME$}{Fact~\ref{fact:freezingdecidable}} & \resref{Undecidable}{Theorem~\ref{thm:basicundecidable}} & \resref{Undecidable}{Theorem~\ref{thm:basicundecidable}}\\
    \hline
    $\PRED{F}$ in 1D & \upperbd{$\NL$}{Proposition~\ref{prop:bclogspace}} & \upperbd{$\NL$}{Proposition~\ref{prop:bclogspace}} & \lowerbd{P-complete}{Proposition~\ref{prop:hard1Dconvergent}}\\
    \hline
    $\PRED{F}$ in 2D & \lowerbd{P-complete}{Proposition~\ref{prop:2Dhardfreezing}} & \lowerbd{P-complete}{Proposition~\ref{prop:2Dhardfreezing}} & \lowerbd{P-complete}{Proposition~\ref{prop:2Dhardfreezing}}\\
    \hline
    $\COM{F}$ in 1D & \upperbd{$O(\log(n))$}{Theorem~\ref{thm:cc}} & \upperbd{$O(\log(n))$}{Theorem~\ref{thm:cc}} & \lowerbd{$\Omega(\sqrt(n))$}{Prop~\ref{prop:highccconvergent}}\\
    \hline
    \begin{tabular}[t]{r}
    $F^\omega(c)$ in 1D\\ $c$ computable
    \end{tabular}& \upperbd{Computable}{Theorem~\ref{thm:computablelimits}} & \upperbd{Computable}{Theorem~\ref{thm:computablelimits}} & \lowerbd{Uncomputable}{Theorem~\ref{thm:uncomputablelimit}}\\
    \hline
    \begin{tabular}[t]{r}
    ${\chi_q(F^\omega(c))}$ in 2D\\ $c$ finite
    \end{tabular}& \lowerbd{Neither r.e. nor co-r.e.}{Corollary~\ref{cor:hardlimit2Dfreezing}} & \lowerbd{Neither r.e. nor co-r.e.}{Corollary~\ref{cor:hardlimit2Dfreezing}} &\lowerbd{Neither r.e. nor co-r.e.}{Corollary~\ref{cor:hardlimit2Dfreezing}} 
  \end{tabular}
\end{center}

\bibliography{refs}
\bibliographystyle{splncs03}

\end{document}